\documentclass[final]{siamart171218}

\usepackage{amsfonts}
\usepackage{graphicx}
\ifpdf
  \DeclareGraphicsExtensions{.eps,.pdf,.png,.jpg}
\else
  \DeclareGraphicsExtensions{.eps}
\fi

\usepackage{dsfont}
\usepackage{cite}
\usepackage{bm}
\usepackage{verbatim}
\usepackage{subfigure}
\usepackage[noend]{algpseudocode}
\usepackage[normalem]{ulem}

\newcommand{\eps}{\epsilon}

\DeclareMathOperator*{\argmin}{argmin}
\DeclareMathOperator*{\argmax}{argmax}

\newcommand{\innp}[1]{\left\langle #1 \right\rangle}

\newcommand{\mA}{\mathbf{A}}
\newcommand{\mI}{\mathbf{I}}
\newcommand{\ones}{\mathds{1}}
\newcommand{\zeros}{\textbf{0}}
\newcommand{\vx}{\mathbf{x}}
\newcommand{\vxh}{\mathbf{\hat{x}}}

\newcommand{\vzh}{\mathbf{\hat{z}}}

\newcommand{\vy}{\mathbf{y}}
\newcommand{\vz}{\mathbf{z}}

\newcommand{\vu}{\mathbf{u}}

\newcommand{\defeq}{\stackrel{\mathrm{\scriptscriptstyle def}}{=}}

\newcommand{\tgrad}{\overline{\nabla f}}
\newcommand{\psih}{\widehat{\psi}}
\newcommand{\tgradj}{\overline{\nabla_j f}}
\newcommand{\vgamma}{\bm{\gamma}}

 \newsiamthm{fact}{Fact}
 \newsiamremark{remark}{Remark}
  \newsiamremark{example}{Example}

\numberwithin{theorem}{section}

\AtBeginDocument{%
  \paperwidth=\dimexpr
    1in + \oddsidemargin
    + \textwidth
    + 1in + \oddsidemargin
  \relax
  \paperheight=\dimexpr
    1in + \topmargin
    + \headheight + \headsep
    + \textheight
    + 1in + \topmargin
  \relax
  \usepackage[pass]{geometry}\relax
}

\newcommand{\TheTitle}{Fair Packing and Covering on a Relative Scale} 
\newcommand{\TheAuthors}{J. Diakonikolas, M. Fazel, and L. Orecchia}

\headers{Fair Packing and Covering on a Relative Scale}{\TheAuthors}

\title{{\TheTitle}\thanks{Part of this work was done while the authors were visiting the Simons Institute for the Theory of Computing. It was partially supported by NSF grants CCF-1718342 and CCF-1409836, NSF TRIPODS Award 1740551, ONR grant N000141612789, and by the DIMACS/Simons Collaboration on Bridging Continuous and Discrete Optimization through NSF grant CCF-1740425.}}

\author{
  Jelena Diakonikolas\thanks{Department of Computer Sciences, UW-Madison, Madison, WI
    (\email{jelena@cs.wisc.edu})}
    \and
    Maryam Fazel\thanks{Department of Electrical Engineering, University of Washington, Seattle, WA
    (\email{mfazel@uw.edu})}
  \and
  Lorenzo Orecchia\thanks{Department of Computer Science, University of Chicago, Chicago, IL
  (\email{orecchia@uchicago.edu})}
}

\usepackage{amsopn}
\DeclareMathOperator{\diag}{diag}

\newif\ifmarkup
\markuptrue

\ifpdf
\hypersetup{
  pdftitle={\TheTitle},
  pdfauthor={\TheAuthors}
}
\fi

\begin{document}
\maketitle

\newcommand{\mf}[1]{{\textcolor{orange}{#1}}}


\begin{abstract}
Fair resource allocation is a fundamental optimization problem with applications in operations research, networking, and economic and game theory. 
Research in these areas 
has led to the general acceptance of 
a class of $\alpha$-fair utility functions parameterized by  $\alpha \in [0, \infty]$. 
We consider $\alpha$-fair packing -- the problem of maximizing $\alpha$-fair utilities under positive linear constraints -- and provide a simple first-order method for solving it with relative-error guarantees. The method has a significantly lower convergence time than the state of the art, and to analyze it, we leverage the Approximate Duality Gap Technique, which provides an intuitive interpretation of the convergence argument. 
Finally, we introduce a natural counterpart of $\alpha$-fairness for minimization problems and motivate its usage in the context of fair task allocation. This generalization yields $\alpha$-fair covering problems, for which we provide the first near-linear-time solvers with relative-error guarantees.
\end{abstract}



\section{Introduction}\label{sec:intro}

How to split 
limited resources is a fundamental question 
studied since  antiquity. 
The study of fairness in economic theory, operations research, and networking led to a single class of utility functions known as $\alpha$-fair utilities~\cite{atkinson1970measurement,MoWalrand2000}:
\begin{equation}\label{eq:alpha-fair-utilities}
f_{\alpha}(x) = 
\begin{cases}
\frac{{x}^{1-\alpha}}{1-\alpha}, &\text{ if } \alpha \geq 0,\, \alpha \neq 1,\\
\log(x), &\text{ if } \alpha = 1.
\end{cases}
\end{equation}
When $\sum_j f_\alpha(x_j)$ is maximized over a convex set, with $x_j$ corresponding to the share of the resource to party $j,$ the resulting solution is equivalent to the  $\alpha$-fair allocation as defined by~\cite{MoWalrand2000}. 
This class of problems is well-studied in the literature, and axiomatically justified in several works~\cite{bertsimas2011price,lan2010axiomatic,joe2013multiresource}. 
Notable special cases of $\alpha$-fair allocations include: (i) utilitarian allocations (with linear objectives), for $\alpha = 0$, (ii) proportionally fair allocations that correspond to Nash bargaining solutions~\cite{nash1950bargaining}, for $\alpha = 1,$ (iii) TCP-fair objectives that correspond to bandwidth allocations in the Internet TCP~\cite{kelly2014stochastic}, for $\alpha = 2,$ and (iv) max-min fair allocations that correspond to Kalai-Smorodinsky solutions in bargaining theory~\cite{kalai1975other}, for $\alpha \rightarrow \infty.$ 

In this paper, we consider the maximization of $\alpha$-fair utilities subject to positive linear (packing) constraints, to which we refer as the $\alpha$-fair packing problems. Given a non-negative matrix {$\mA\in\mathbf{R_+}^{m\times n}$} 
and a parameter $\alpha \geq 0$, they are defined as~\cite{marasevic2015fast}:
\begin{equation}\tag{P-a}\label{eq:alpha-fair-packing}
\begin{aligned}
\max \Big\{f_{\alpha}(\vx) \defeq \sum_{j=1}^n f_{\alpha}(x_j):\quad & \mA\vx\leq \ones,\;\vx \geq \zeros\Big\},
\end{aligned}
\end{equation}
where {$\vx\in\mathbf{R}^n$},  $\ones$ is an all-ones vector, $\zeros$ is an all-zeros vector. Packing constraints are natural in many applications, including Internet congestion control~\cite{low2002internet}, rate control in software defined networks~\cite{mccormick2014real}, 
multi-resource allocation in data centers~\cite{bonald2015multi,joe2013multiresource,ghodsi2011dominant},
and a variety of applications in operations research, economics, and game theory~\cite{bertsimas2012efficiency,jain2007eisenberg}. When $\alpha = 0,$ the problem is equivalent to a packing linear program (LP). 

We are interested in solving~\eqref{eq:alpha-fair-packing} in a \emph{distributed} model of computation, where each coordinate $j$ of the allocation vector $\vx$ is updated according to global problem parameters (e.g., $m, n$), local information for coordinate $j$ (namely, the $j^{\mathrm{th}}$ column of $\mA$), and local{\footnote{{There are two reasons why this information should be considered local. The first is that we can, in many situations, view the constraints as distributed agents, in which case the locality of information is clear. The second is that algorithms for fractional packing and covering problems working in the same model as presented here are considered local even in the distributed computing community (see, e.g.,~\cite{d-kuhn2006price}). A classical example is fractional matching on graphs, which is widely used as a model for scheduling in wireless networks under interference constraints.}}} information received in each round. The local per-round information for coordinate $j$ is the slack $1-(\mA\vx)_i$ of all the constraints $i$ in which $j$ participates. Such a model is natural for networking applications{, where each variable $x_j$ corresponds to the rate assigned to a route $j$ and $1-(\mA\vx)_i$ are the (relative) congestion on each of the links $i$ that belong to the route $j$} (see, e.g., the textbook~\cite{kelly2014stochastic} and references therein). Further, as resource allocation problems frequently arise in large-scale settings in which results must be provided in real-time (e.g., in data center resource allocation~\cite{ghodsi2011dominant,joe2013multiresource,bonald2015multi,jose2019distributed}), the design of distributed solvers that efficiently compute approximate  solutions to $\alpha$-fair allocation problems is of crucial importance. 

First-order methods are particularly relevant in this context as they lead to algorithms that can be distributed, have simple updates implementable in large-scale settings, and are efficient in practice. Further, we focus on algorithms that are \emph{width-independent}\footnote{I.e., their iteration complexities scale poly-logarithmically with the matrix width, defined as the maximum ratio of $\mA$'s non-zero entries.} and yield an $\epsilon$-approximate solution in the sense of \emph{relative error}. 
Width-independent algorithms are of great theoretical interest: algorithms that are not width-independent cannot in general be considered polynomial-time. Such algorithms have primarily been studied in the context of packing and covering LPs. 
From an optimization perspective, width-independence is surprising, as black-box application of any of the standard first-order methods \emph{does not lead to width-independent algorithms}. We also note that methods with relative-error guarantees are much less studied in optimization than their additive-error counterparts, and are typically confined to positive linear programs (see, e.g.,~\cite[Chapter 7]{nesterov2018lectures} and references therein). 

Finally, we note that for $\alpha > 0,$ $\alpha$-fair utilities do not possess any of the global regularity properties such as smoothness or Lipschitz-continuity that are typically used in convergence analysis of first-order methods. In fact, as any of the coordinates $x_j$ tends to zero, $\nabla_j f_\alpha(x_j) \to \infty.$ Notably, it is possible to make the Lipschitz constant of the objective or its gradient finite by enforcing $x_j \geq \delta,$ for a sufficiently small $\delta.$ However, to ensure that the feasible region contains an $\epsilon$-approximate solution to~\eqref{eq:alpha-fair-packing}, it is necessary that $\delta \leq \frac{1}{n \max_{ij}A_{ij}},$ which leads to a prohibitively large Lipschitz constant $\Omega((n \max_{ij}A_{ij})^\alpha)$ of the objective {and $\Omega((n \max_{ij}A_{ij})^{\alpha+1})$ of the gradient,} and thus algorithms that are not width-independent. {In particular, if one was to use Nesterov's ``smooth minimization of nonsmooth functions''~\cite{nesterov2005smooth} to deal with the constraints, the resulting number of iterations required to construct a solution with the same approximation guarantee as in this paper would be no better than $\frac{\rho((n \rho)^{\alpha+1} + \rho)}{\epsilon}$, where $\rho = \frac{\max_{ij} A_{ij}}{\min_{kl:A_{kl}\neq 0} A_{kl}}$ is the matrix width.} {Alternatively, because $\alpha$-fair objectives are $\alpha$-strongly convex, it is possible to work with the dual problem, which is smooth (gradient-Lipschitz) for $\alpha>0.$ This idea was pursued in~\cite{Beck2014Gradient}. However, the smoothness constant is proportional to the maximum eigenvalue of $\mA^T\mA$ times $1/\alpha$, which is not width-independent, and in the worst case scales as $mn(\max_{ij}A_{ij})^2/\alpha,$ leading to the overall polynomial dependence on $m, n, \rho$ and linear dependence on $1/\epsilon,$ similar to the approach described above.} These issues are handled in our analysis by using a more fine-grained smoothness-like property of the objective that only holds locally and with sufficiently small step sizes (see Lemma~\ref{lemma:smoothness}).

\subsection{Contributions}

We obtain improved distributed algorithms for constructing $\epsilon$-approximate\footnote{As in~\cite{marasevic2015fast}, the approximation is multiplicative for $\alpha\neq 1$ and additive for $\alpha = 1$.} solutions to $\alpha$-fair packing problems. As in~\cite{marasevic2015fast}, our specific convergence results depend on the regime of the parameter $\alpha$, where each iteration takes linear work in the number of non-zero elements of $\mA$. 
\begin{itemize}
\item For $\alpha \in [0,1)$, Theorem~\ref{thm:a=01convergence} shows that a solution with $(1+\epsilon)$-relative error is reached within  $O\big(\frac{\log(n\rho)\log(mn\rho/\epsilon)}{(1-\alpha)^3\epsilon^2}\big)$ iterations. This bound  matches the best known results for parallel packing LP solvers for $\alpha = 0$~\cite{d-allen2014using,d-wang2015unified}, and improves the dependence on $\eps$ compared to~\cite{marasevic2015fast} from $\eps^{-5}$ to $\eps^{-2}$ for $\alpha \in (0, 1)$.
\item For $\alpha = 1,$ Theorem~\ref{thm:a=1-convergence} yields $\eps$-approximate convergence  in $ O\big(\frac{\log^3(m n\rho/\epsilon)}{\epsilon^2}\big)$ iterations. In this case only, the error is additive, as the objective can take both positive and negative values. 
The dependence on $\eps$ compared to~\cite{marasevic2015fast} is improved from $\eps^{-5}$ to $\eps^{-2}$.
\item For $\alpha > 1,$ Theorem~\ref{thm:a>1-convergence} shows that a solution with $(1-\epsilon)$-relative error\footnote{The relative error in this case is $1-\epsilon$, because for $\alpha>1$ the objective functions are negative.} is reached within  $O\big(\max\big\{\frac{\alpha^3\log(n\rho/\eps)\log(mn\rho/\epsilon)}{\epsilon}, \frac{\log(\frac{1}{\epsilon(\alpha-1)})\log(mn\rho/\epsilon)}{\epsilon(\alpha-1)}\big\}\big)$
iterations. This can be extended to the $\max$-$\min$-fair case for sufficiently large $\alpha$ ~\cite{marasevic2015fast}. The dependence on $\eps$ compared to~\cite{marasevic2015fast} is improved from $\eps^{-4}$ to $\eps^{-1}$.
\end{itemize}

While the analysis for each of these cases is somewhat involved, the algorithms we propose are extremely simple, as described in Algorithm~\ref{algo:fair-pc} of Section~\ref{se:algo}. Moreover, our dependence on $\eps$ is improved by a factor $\eps^{-3}$ {(the dependence on all the remaining parameters is either the same or improved)} and the analysis is simpler than the one from~\cite{marasevic2015fast}, as it leverages the Approximate Duality Gap Technique (ADGT)~\cite{thegaptechnique}. 

Our final contribution is to introduce a natural counterpart of $\alpha$-fairness for minimization problems, which we use to study $\beta$-fair covering problems{, for $\beta \geq 0$}\footnote{We use $\beta$ instead of $\alpha$ to distinguish between the different parameters in the convergence analysis.}:
\begin{equation}\tag{C}\label{eq:beta-fair-covering}
\begin{aligned}
\min \Big\{ g_{\beta}(\vy) \defeq \sum_{i=1}^m \frac{{y_i}^{1+\beta}}{1+\beta}:\quad & \mA^T \vy \geq \ones,\; \vy \geq \zeros\Big\}.
\end{aligned}
\end{equation}
As for packing problems, the $\beta$-covering formulation can be motivated by the goal of producing an equitable allocation of cost among different agents. For instance, we may want to allocate work hours to different workers so that various production requirements, given as covering constraints, are met. In this case, assigning all work to one worker may provide a solution that minimizes total work but unfairly singles out this worker.  
However, a fair solution would allocate work so that no worker gets assigned too much work and every worker performs some portion of the work. This is captured in~\eqref{eq:beta-fair-covering} by the fact that the objective quickly grows to infinity for $\beta > 0,$ as the amount of work $y_i$ given to worker $i$ increases. 
This generalization yields $\beta$-fair covering problems, for which we provide the first width-independent $\epsilon$-approximate solver that converges in $O(\frac{(1+\beta)\log(mn\rho)}{\beta\epsilon})$ iterations, by reducing the  analysis to the $\alpha < 1$ case of the $\alpha$-fair packing problems (Section~\ref{se:covering}).

\subsection{{Our Techniques}} 

{
Several difficulties arise when considering cases $\alpha > 0$ compared to the linear case ($\alpha = 0$). Unlike linear objectives which are $1$-Lipschitz, $\alpha$-fair objectives for $\alpha > 0$ lack any good global properties typically used in convex optimization, 
e.g., Lipschitz-continuity of the function or its gradient.  
As mentioned before, it is possible to prune the feasible region to guarantee positivity of the vector $\vx$. However, any pruning that retains $\eps$-approximate solutions would require the point $\frac{1}{n\rho}\ones$ to be in the pruned set, leading to Lipschitz constants of order $\left({n\rho}\right)^\alpha$ and $(n\rho)^{\alpha + 1}.$ This makes it hard to directly apply arguments relying on gradient truncation used in the packing LP case~\cite{d-allen2014using}.  
}

{
To circumvent this issue, we use a change of variable, which reduces the objective to a linear one, but makes the constraints more complicated. Further, in the case $\alpha > 1,$ the truncated gradient has the opposite sign compared to the $\alpha \leq 1$ cases. Though this change in the sign may seem minor, it invalidates the arguments that are typically used in analyzing distributed packing LP solvers~\cite{d-allen2014using,LP-jelena-lorenzo}, which is one of the main reasons why in the linear case the solution to the covering LP is obtained by solving its dual -- a packing LP. Unfortunately, in the case $\alpha > 1,$ solving the dual problem seems no easier than solving the primal -- in terms of truncation, the gradients have the same structure as in the covering LP.}

{
Similar to the linear case~\cite{d-allen2014using}, we use regularization of the constraint set to turn the problem into an unconstrained optimization problem over the non-negative orthant. The regularizing function is different from the standard generalized entropy typically used for LPs, and belongs to the same class of functions considered in the fair covering problem. 
 These regularizers seem more natural than entropy, as local smoothness properties used in  algorithm analysis hold regardless of whether the point at which local smoothness is considered satisfies the packing constraints, which is not true for entropic regularization. 
Furthermore, these regularizers are crucial for reducing fair covering problems to $\alpha$-fair packing problems with $\alpha < 1$ (see Section~\ref{se:covering}).
}

{
While the analysis of the case $\alpha \in [0, 1)$ is similar to the analysis of packing LPs from the unpublished note~\cite{LP-jelena-lorenzo} by a subset of the authors, it is not clear how to generalize this argument to the cases $\alpha = 1$ and $\alpha > 1,$ with these techniques or any others developed for packing LPs (see Sections~\ref{sec:fair-pc} and~\ref{se:algo} for more details). The analysis of the case $\alpha = 1$ is relatively simple, and can be seen as a generalization of the gradient descent analysis. }

{
However, the case $\alpha > 1$ is significantly more challenging. First, ADGT~\cite{thegaptechnique} cannot be applied directly, for a number of reasons: (i) it is hard to argue that the optimality gap of any naturally chosen initial solution is a constant-approximation-factor away from the optimum; this is because when $\alpha > 1,$ $f_\alpha(\frac{1}{n\rho}) = -\frac{(n\rho)^{\alpha-1}}{\alpha-1}$, which may be much smaller than $-n/(\alpha-1)$, while the optimal solution can be as large as $-{n}/({\alpha-1}),$ (ii) gradient truncation cannot be applied to the approximate gap constructed by ADGT (see Section~\ref{se:algo}), and (iii) without the gradient truncation, it is not clear how to argue that the approximate gap from ADGT decreases at the right rate (or at all), which is crucial for the ADGT argument to apply. 
One of the reasons for (iii) is that the approximate dual in ADGT can be a very crude approximation of the true Lagrangian dual for $\alpha > 1.$ }

{
Our main idea is to use the Lagrangian dual of the original, non-regularized problem, with two different arguments. The first argument is \emph{local} and relies on similar ideas as~\cite{marasevic2015fast}: it uses only the current iterate to argue that if certain regularity conditions do not hold, the regularized objective must decrease by a sufficiently large multiplicative factor. Compared to~\cite{marasevic2015fast}, we require much looser regularity conditions, which eventually leads to a much better dependence of the convergence time on $\eps$: the dependence is reduced from $1/\eps^4$ to $1/\eps$ without incurring any additional logarithmic factors in the input size, and even improving the dependence on $\alpha.$ This is achieved through the use of the second argument, which relies on the \emph{aggregate history} of the iterates that satisfy the regularity conditions. This argument is more similar to ADGT, though as noted above and unlike in the standard ADGT~\cite{thegaptechnique}, the approximate gap is constructed from the Lagrangian dual of the original problem. To show that the approximate gap decreases at the right rate, we use a careful coupling of the regularity conditions with the gradient truncation (see Section~\ref{sec:a>1-gap-decrease}).}
 
\subsection{{Additional} Related Work}
A long line of work on packing and covering LPs has resulted in width-independent distributed algorithms {(see}~\cite{d-kuhn2006price,d-luby1993parallel,d-bartal1997global,dc-young2001sequential,d-papadimitriou1993linear,AwerbuchKhandekar2009} {and references therein)}. This has culminated in recent results that ensure convergence to an $\epsilon$-approximate optimal solution in $O(\log(n) \log(mn/\epsilon)/ \eps^2) $ rounds of computation~\cite{d-allen2014using,d-wang2015unified}. However, when it comes to the general $\alpha$-fair resource allocation, only~\cite{marasevic2015fast} provides a width-independent algorithm. The algorithm of~\cite{marasevic2015fast} works in a very restrictive setting of stateless distributed computation, which leads to convergence times that are poly-logarithmic in the problem parameters, but have high dependence on the error parameter $\epsilon$ (namely, the dependence is $\epsilon^{-5}$ for $\alpha \leq 1$ and $\eps^{-4}$ for $\alpha > 1$).%

{The stateless model of distributed computation requires the algorithm (i) to be able to start from an arbitrary (not necessarily feasible) point, (ii) to not have memory of previous states (previous solution points; except for the last one), and (iii) to work in an asynchronous setting, where individual coordinates can be updated based on stale information. While our algorithm satisfies Property (ii) and can be adjusted to satisfy Property (i), it does not satisfy the third property, which would force the step sizes to be smaller by a factor $\epsilon$ and lead to overall slower convergence. Note that better convergence bounds achieved in this work \emph{cannot}  be obtained by simply taking a larger step size  and following the same analysis as in~\cite{marasevic2015fast}. This is because the step sizes from~\cite{marasevic2015fast} are also crucially used in the analysis to guarantee algorithm progress. What leads to faster convergence in our work is the use of ADGT, which allows us to adapt the step sizes to the gradient (step sizes in~\cite{marasevic2015fast} are not adaptive) and work with the \emph{aggregate} information over \emph{all} iterations, thus constructing better dual solutions and having a global view of the problem. By contrast, the argument from~\cite{marasevic2015fast} is \emph{local}, and only uses information from the last iteration to guarantee the algorithm progress or argue that the algorithm has converged to an $\epsilon$-approximate solution.}

\subsection{Organization of the Paper} Section~\ref{sec:prelims} introduces 
the necessary notation and background. Section~\ref{se:algo} provides the statement of the algorithm for $\alpha$-fair packing and overviews the main technical ideas. The full technical argument is provided in Section~\ref{sec:proofs}.  Section~\ref{se:covering} provides the results for $\beta$-fair covering. We conclude in Section~\ref{sec:conclusion}.

\section{Preliminaries}\label{sec:prelims}

We assume that the problems are expressed in their standard scaled form~\cite{marasevic2015fast,d-luby1993parallel,AwerbuchKhandekar2009,d-allen2014using}, so that the minimum non-zero entry of the constraint matrix $\mA$ equals one and the maximum element of $\mA$ is equal to the matrix width $\rho$.
{Note that 
weighted versions of the problems, 
with objective $\sum_{j}w_j f_\alpha(x_j)$ for positive weights $w_j$, can be expressed in this form through rescaling and the change of variable. The only effect on the final bounds is that $\rho$ would also depend on $\frac{\max_j w_j}{\min_k w_k},$ which only appears under a logarithm in our bounds, similar to~\cite{marasevic2015fast}.} 
{The scaling is necessary only for the analysis; the algorithm can be stated for the non-scaled problem by reverting the change of variable (see~\cite{marasevic2015fast,awerbuch2008fast,d-allen2014using} for similar arguments).} The constraint matrix $\mA$ is $m \times n$; $\mI$ denotes the identity matrix.

\subsection{Notation and Useful Definitions and Facts}

We use boldface letters to denote vectors and matrices, and italic letters to denote scalars. We let $\vx^a$ denote the vector $[{x_1}^a, {x_2}^a, ..., {x_n}^a]^T$, $\exp(\vx)$ denote the vector $[\exp(x_1), \exp(x_2),...,\exp(x_n)]^T$. Inner product of two vectors is denoted as $\innp{\cdot, \cdot}$, while the matrix/vector transpose is denoted by $(\cdot)^T$. 
$\nabla_j f(\cdot)$ denotes the $j^\mathrm{th}$ coordinate of $\nabla f$, i.e., $\frac{\partial f}{\partial x_j}$.
We use the following notation for the truncated (and scaled) gradient~\cite{d-allen2014using}, for $\alpha \neq 1$:
\begin{equation}\label{eq:trunc-grad}
\tgradj (\vx) = 
\begin{cases}
(1-\alpha)\nabla_j f(\vx), \text{ if } & (1-\alpha)\nabla_j f(\vx) \in [-1, 1],\\
1,  \text{ otherwise. }
\end{cases}
\end{equation}
As we will see later, the only relevant case for us will be the functions whose gradient coordinates satisfy  $(1-\alpha)\nabla_j f(\vx) \in [-1, \infty)$. Hence, the gradient truncation is irrelevant for $(1-\alpha)\nabla_j f(\vx) < -1$.
The definition of the truncated gradient for $\alpha=1$ is equivalent to the definition~\eqref{eq:trunc-grad} with $\alpha=0.$

Most functions we will work with are convex differentiable functions defined on $\mathbb{R}^n_+$. Thus, we will be stating all definitions assuming that the functions are defined on $\mathbb{R}_+^n$. A useful definition of convexity of a (differentiable) function $f:\mathbb{R}^n_+ \rightarrow \mathbb{R}$ is:
\begin{equation}\label{eq:1-o-cvxity}
f(\vy) \geq f(\vx) + \innp{\nabla f(\vx),\, \vy - \vx}, \; \forall \vx, \vy \in \mathbb{R}^n_+.
\end{equation}

Convex conjugate of a function $\psi: \mathbb{R}_+^n\rightarrow \mathbb{R}$ is defined as:
$$\psi^*(\vz) \defeq \sup_{\vx \geq \zeros}\{\innp{\vz, \vx} - \psi(\vx)\}.$$ %
In all the examples we consider in this paper, $\sup$ can be replaced by $\max$. 
The following standard fact about 
convex conjugates can be obtained as a corollary of Danskin's theorem (see, e.g.,~\cite{bertsekas1971control}).
\begin{fact}\label{fact:cvx-conc-conj}
Convex conjugate $\psi^*$ of a function $\psi$ is convex. %
Moreover, if $\psi$ is strictly convex,  $\psi^*$ is differentiable, 
and the following holds:  
$$
\nabla \psi^*(\vz) = \argmax_{\vx \geq \zeros}\{\innp{\vz, \vx} - \psi(\vx)\}.
$$
\end{fact}
\subsection{Fair Packing and Covering}\label{sec:fair-pc}
Recall that $\alpha$-fair packing problems were defined by~\eqref{eq:alpha-fair-packing}. 
In the analysis, there are three regimes of $\alpha$ that are handled separately: $\alpha \in [0, 1),$ $\alpha =1,$ and $\alpha > 1.$ In these three regimes, the $\alpha$-fair utilities $f_\alpha$ exhibit very different behaviors, as illustrated in Fig~\ref{fig:a-fair-illustration}. {When $\alpha = 0,$ $f_\alpha$ is just the linear utility, and~\eqref{eq:alpha-fair-packing} is a packing LP. As $\alpha$ increases from zero to one, $f_\alpha$ remains non-negative, but its shape approaches the shape of the natural logarithm. In this regime, any feasible solution $\vx$ has optimality gap that is bounded by a constant multiple of $f_\alpha(\vx^*),$ where $\vx^*$ is the solution to~\eqref{eq:alpha-fair-packing}, as $f_\alpha(\vx^*) - f_\alpha(\vx) \leq f_\alpha(\vx^*).$ Thus, in this regime, any algorithm working within the feasible space of~\eqref{eq:alpha-fair-packing} only needs to reduce the optimality gap by a factor $1/\epsilon.$}

{When $\alpha =1,$ $f_\alpha$ is the natural logarithm. Even though the range of the possible values of $f_\alpha$ within the feasible set of~\eqref{eq:alpha-fair-packing} is infinite in this case, it is not hard to choose an initial solution $\vx_0$ such that $f_\alpha(\vx^*) - f_\alpha(\vx_0) \leq n\log(n\rho)$ (see Proposition~\ref{prop:packing-bounds-on-opt}). This suffices for our analysis, as we only aim for the final error of the order $n \epsilon.$ Notice also that in this case we can only hope for additive error, as the logarithmic function takes both positive and negative values. }

{When $\alpha > 1,$ $f_\alpha$ is non-positive and its shape approaches the shape of the natural logarithm as $\alpha\to 1$. As $\alpha$ increases and approaches infinity, $f_\alpha$ bends and becomes steeper, approaching 
the negative indicator of the interval $[0, 1]$.  One of the challenges that occurs in the analysis is that it is unclear how to choose an initial feasible point $\vx_0$  whose optimality gap would be a constant or poly-log factor of $f_\alpha(\vx^*),$ as $|f_\alpha(\vx_0)| \geq |f_\alpha(\vx^*)|$ for any feasible $\vx_0,$ and any reasonable $\vx_0$ could have an optimality gap that is an order $(n\rho)^{\alpha-1}$-factor away from $f_\alpha(\vx^*).$ This crucially affects the analysis, as discussed in Section~\ref{sec:adgt}.}
\begin{figure}
\centering
\hspace{\fill}
\subfigure[$\alpha \in [0, 1)$]{\includegraphics[width=0.27\textwidth]{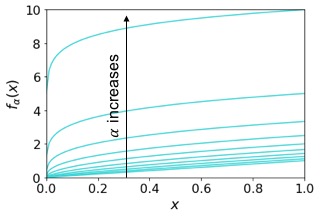}\label{fig:a<1}}
\hspace{\fill}
\subfigure[$\alpha = 1$]{\includegraphics[width=0.27\textwidth]{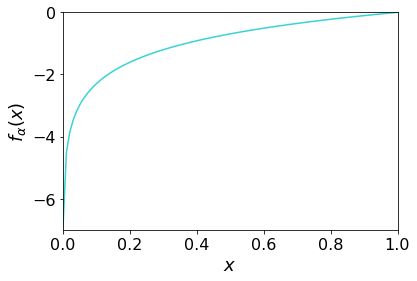}\label{fig:a=1}}
\hspace{\fill}
\subfigure[$\alpha >1$]{\includegraphics[width=0.27\textwidth]{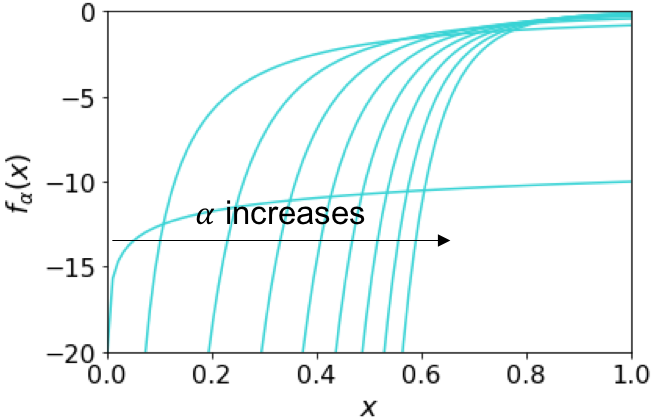}
\label{fig:a>1}}
\hspace*{\fill}
\caption{The three regimes of $\alpha$: \protect\subref{fig:a<1} $\alpha \in [0, 1)$, where $f_\alpha$ is non-negative and equal to zero at $x = 0$; for $\alpha = 0$, $f_\alpha$ is linear, and as $\alpha$ approaches 1 from below, $f_\alpha$ approaches the shifted logarithmic function that equals zero at $x = 0$; \protect\subref{fig:a=1} $\alpha = 1$, where $f_\alpha$ is the natural logarithm; and \protect\subref{fig:a>1} $\alpha > 1,$  where $f_\alpha$ is non-positive; when $\alpha$ approaches 1 from above, $f_\alpha$ approaches the shifted logarithmic function that tends to zero as $x$ tends to $\infty$; when $\alpha \rightarrow \infty,$  $f_\alpha$ approaches the step function that is $-\infty$ for $x \in (0, 1)$ and zero for $x \geq 1.$}
\label{fig:a-fair-illustration}
\end{figure}

Consider the following change of variable:
\begin{equation}
\vx = F_\alpha(\vxh) \defeq 
\begin{cases}
\vxh^{\frac{1}{1-\alpha}}, &\text{ if } \alpha \geq 0,\; \alpha \neq 1,\\
\exp(\vxh), &\text{ if } \alpha = 1.
\end{cases}
\end{equation}
{This change of variable does not affect the algorithm; it is only used for the convenience of the analysis.} 
Let $S_{\alpha} = \mathbb{R}_+^n$, $\hat{f}_{\alpha}(\vx) = \frac{\innp{\ones, \vx}}{1-\alpha}$ for $\alpha \neq 1, \alpha \geq 0$ and $S_{\alpha} = \mathbb{R}^n$, $\hat{f}_{\alpha}(\vx) = \innp{\ones, \vx}$ for $\alpha =1$. The problem \eqref{eq:alpha-fair-packing} can then equivalently be written (with the abuse of notation) as:
\begin{equation}\label{eq:packing-transformed}\tag{P-b}
\max \big\{\hat{f}_\alpha(\vx):\quad  \mA F_\alpha(\vx) \leq \ones,\; \vx \in S_\alpha\big\}.
\end{equation}
{Observe that there is one-to-one correspondence between $\vx$ and $F_\alpha(\vx):$ $\vx$ is \eqref{eq:packing-transformed}-feasible if and only if $F_\alpha(\vx)$ is \eqref{eq:alpha-fair-packing}-feasible. The objective function value also remains the same: $f_\alpha(F_\alpha(\vx)) = \hat{f}_\alpha(\vx)$. Thus, any statements we make about \eqref{eq:packing-transformed} can be translated into statements about \eqref{eq:alpha-fair-packing}; this will be used repeatedly in the analysis.}

To bound the optimality gap in the analysis, it is important to bound the optimum objective function values, as stated in the following proposition. 
\begin{proposition}\label{prop:packing-bounds-on-opt}
Let $\vx^*$ be (any) optimal solution to  \eqref{eq:packing-transformed}. Then:
\begin{itemize}
\item If $\alpha \geq 0$ and $\alpha \neq 1,$ $\frac{n}{1-\alpha}(n\rho)^{\alpha-1} \leq \hat{f}_\alpha(\vx^*)\leq \frac{n}{1-\alpha}$.
\item If $\alpha = 1$, $-n\log(n \rho)\leq \hat{f}_\alpha(\vx^*)\leq 0$.
\end{itemize}
\end{proposition}
\begin{proof}
The proof is based on the following simple argument. When $F_\alpha(\vx) = \frac{1}{n\rho} \ones,$ $\vx$ is feasible and we get a lower bound on the optimal objective value. On the other hand, if $F_\alpha(\vx) > \ones$, then (as the minimum non-zero entry of $\mA$ is at least 1), all constraints are violated, which gives an upper bound on the optimal objective value. 
\end{proof}

Write \eqref{eq:packing-transformed} as the following saddle-point problem:
\begin{equation}\label{eq:p-saddle-point}
\min_{\vx \in S_\alpha} -\hat{f}_\alpha(\vx) + \max_{\vy\geq 0} \innp{\mA F_\alpha(\vx) - \ones, \vy}.
\end{equation}
The main reason for considering the saddle-point formulation of \eqref{eq:packing-transformed} is that after regularization it can be turned into an unconstrained problem over the positive orthant, without losing much in the approximation error, under mild regularity conditions on the steps of the algorithm. In particular, let $(\vx^*, \vy^*)$ be the optimal primal-dual pair in~\eqref{eq:p-saddle-point}. Then, by Fenchel duality (see, e.g.,~\cite[Proposition 5.3.8]{bertsekas2009convex}), we have that $-\hat{f}_\alpha(\vx^*) = \min_{\vx\in S_\alpha} \{-\hat{f}_\alpha(\vx) + \innp{\mA F_\alpha(\vx) - \ones, \vy^*}\}$. Hence, $\forall \vx \geq \zeros$:
\begin{equation}\label{eq:regularization}
\begin{aligned}
-\hat{f}_\alpha(\vx^*) &\leq -\hat{f}_\alpha(\vx) + \innp{\mA F_\alpha(\vx) - \ones, \vy^*} - \psi(\vy^*) + \psi(\vy^*)\\
&\leq -\hat{f}_\alpha(\vx) + \max_{\vy \geq \zeros}\{\innp{\mA F_\alpha(\vx) - \ones, \vy} - \psi(\vy)\}+ \psi(\vy^*)\\
&= {f_r}(\vx) + \psi(\vy^*), 
\end{aligned}
\end{equation} 
where $f_r(\vx) = -\hat{f}_\alpha(\vx) + \psi^*(\mA F_\alpha(\vx) - \ones)$. {Function $\psi^*(\mA F_\alpha(\vx) - \ones)$ can also be viewed as an approximate barrier for the packing polytope. }
The main idea is to show that we can choose the function $\psi$ so that ${f_r}$ closely approximates $-f_\alpha$ around the optimum $\vx^*$, and, further, we can recover a $(1 + O(\epsilon))$-approximate solution to \eqref{eq:alpha-fair-packing} from a $(1+\eps)$-approximate solution to $\min_{\vx \geq \zeros} {f_r}(\vx)$. This will allow us to focus on the minimization of ${f_r}$, without the need to worry about satisfying the packing constraints from \eqref{eq:alpha-fair-packing} in each iteration. The following proposition formalizes this statement and introduces the missing parameters. Its proof is provided in Appendix~\ref{app:omitted-proofs-prelims}. In the choice of $\psi(\cdot)$, the factor $C^{-\beta}$ ensures that the algorithm maintains (strict) feasibility. The case $C=1$ would allow violations of the constraints by a factor $(1+\epsilon)$.
\begin{proposition}\label{prop:regularization}
Let $\psi(\vy) = \sum_{i=1}^m\big(\frac{{y_i}^{1+\beta}}{C^{\beta}(1+\beta)} - y_i\big),$ where $\beta = \frac{\epsilon/4}{(1+\alpha)\log({4mn\rho}/{\epsilon})},$ $C = (1+\epsilon/2)^{1/\beta},$ and $\epsilon \in \big(0, \min\big\{\frac{1}{2}, \frac{1}{10|\alpha-1|}\big\}\big)$ is the approximation parameter. Then:
\begin{enumerate}
\item $f_r(\vx) = - \hat{f}_\alpha(\vx) + \frac{C\beta}{1+\beta}\sum_{i=1}^m (\mA F_\alpha(\vx))_i^{\frac{1+\beta}{\beta}}.$
\item Let $\vx_r^* = \argmin_{\vx \in S_\alpha}f_r(\vx)$, $\vx^*_\alpha$ be a solution to \eqref{eq:alpha-fair-packing}, and $\vxh_r = F_\alpha(\vx^*_r).$ Then $\vxh_r$ is \eqref{eq:alpha-fair-packing}-feasible and: 
\begin{equation}
-f_\alpha(\vxh_r) + f_\alpha(\vx^*_\alpha) \leq f_r(\vx^*_r) + f_\alpha(\vx^*_\alpha) \leq 2\epsilon_f \defeq
2\begin{cases}
\epsilon n, &\text{ if } \alpha = 1;\\
\epsilon(1-\alpha)f_\alpha(\vx^*_\alpha), &\text{ if } \alpha \neq 1.
\end{cases}\notag
\end{equation}
\end{enumerate}
\end{proposition}

A natural counterpart to $\alpha$-fair packing problems is the $\beta$-fair covering, defined in~\eqref{eq:beta-fair-covering}. %
Similar as in the case of $\alpha$-fair packing, when $\beta = 0$, the problem reduces to the covering LP. It is not hard to show (using similar arguments as in~\cite{MoWalrand2000}) that when $\beta\rightarrow \infty$, the optimal solutions to \eqref{eq:beta-fair-covering} converge to the min-max fair allocation. 

For our analysis, it is useful to work with the Lagrangian dual of \eqref{eq:beta-fair-covering}, given by:
\begin{equation*}
\max_{\vx \geq \zeros} \innp{\ones, \vx} - \frac{\beta}{1+\beta} \sum_{i=1}^m (\mA \vx)_i^{{(1+\beta)}/{\beta}}. 
\end{equation*}
In particular, solving the dual of \eqref{eq:beta-fair-covering} is the same as minimizing ${f}_r(\vx)$ from the packing problem, with $\alpha = 0$ and $\beta$ from the fair covering formulation  \eqref{eq:beta-fair-covering}. 

The following two (simple) propositions will be useful in our analysis.
\begin{proposition}\label{prop:covering-opt-value}
Let $\vy^*$ be an optimal solution to~\eqref{eq:beta-fair-covering}. Then: 
$$
\Big(\frac{1}{m\rho}\Big)^{1+\beta}\frac{m}{1+\beta}\leq \sum_{i=1}^m \frac{(y_i^*)^{1+\beta}}{1+\beta}\leq \frac{m}{1+\beta}.
$$
\end{proposition}
\begin{proposition}\label{prop:covering-bnd-for-x*b}
Let $(\vy_\beta^*, \vx_\beta^*)$ be the optimal primal-dual pair for~\eqref{eq:beta-fair-covering}. Then: 
$$
\innp{\ones, \vx_\beta^*}=(1+\beta)g_\beta(\vy^*_\beta).
$$
\end{proposition}
\begin{proof}
By strong duality,  $\langle{\ones, \vx^*_\beta}\rangle - \frac{\beta}{1+\beta}\sum_{i=1}^m (\mA\vx^*_\beta)_i^{{(1+\beta)}/{\beta}}=g_\beta(\vy^*_\beta)$ and $\vy^*_\beta = (\mA\vx^*_\beta)^{1/\beta}.$ Combining these two identities completes the proof.
\end{proof}
\section{Fair Packing: Algorithm and Convergence Analysis Overview}\label{se:algo}

The algorithm pseudocode is provided in Algorithm~\ref{algo:fair-pc} (\textsc{FairPacking}). All parameter choices will become clear from the analysis.

{Observe that Algorithm~\ref{algo:fair-pc} can be implemented in the distributed model described in the introduction, as all that is needed for each distributed agent $j$ are: (i) global problem parameters $m, n, \rho, \alpha$, and $\epsilon;$ and (ii) the slack $(\mA F_\alpha(\vx^{(k-1)}) - \ones)_i$ for all constraints in which $j$ participates (i.e., for $i$ such that $A_{ij} \neq 0$), in each iteration $k,$ as this information suffices for computing the $j^\mathrm{th}$ coordinate of the truncated gradient, which in turn suffices for computing the new state $x^{(k)}_j.$}
\begin{algorithm}
\caption{\textsc{FairPacking}($\mA, \eps, \alpha$)}\label{algo:fair-pc}
\begin{algorithmic}[1]
\State $\vx^{(0)} = \big(\frac{1-\epsilon}{n\rho}\big)^{1-\alpha}\ones$ for $\alpha \neq 1,$ $\vx^{(0)} = \exp\big(\frac{1-\epsilon}{n\rho}\big)\ones$ for $\alpha = 1,$ $\beta = \frac{\epsilon/4}{(1+\alpha)\log({4mn\rho}/{\epsilon})}$
\If{$\alpha < 1$}
\State $\vz^{(0)} = \exp(\epsilon/4)\ones,$ $\beta' = \frac{(1-\alpha)\epsilon/4}{\log(n\rho/(1-\epsilon))}$, $h = \frac{(1-\alpha)\beta\beta'}{16\epsilon(1+\alpha\beta)}$
\For{$k=1$ to $K=\lceil 2/((1-\alpha)h\epsilon) \rceil$}
\State $\vx^{(k)} = (\ones + \vz^{(k-1)})^{-1/\beta'}$
\State $\vz^{(k)} = \vz^{(k-1)} + \epsilon h \tgrad_r(\vx^{(k)})$
\EndFor
\ElsIf{$\alpha = 1$}
\For{$k=1$ to $K=\left\lceil 10 \frac{\log^2(8\rho m n/\epsilon)}{\epsilon \beta}\right\rceil$}
\State $\vx^{(k)} = \vx^{(k-1)} - \frac{\beta}{4(1+\beta)}\tgrad_r(\vx^{(k-1)})$
\EndFor
\Else 
\For{$k=1$ to $K=\left\lceil 800\frac{(1+\alpha)^2\log(n\rho/(\eps\min\{\alpha-1, 1\}))}{\beta\min\{\alpha-1, 1\}}\right\rceil$}
\State $\vx^{(k)} = (\mI-\frac{\beta(1-\alpha)\diag(\tgrad_r(\vx^{(k-1)})}{4(1+\alpha\beta))})\vx^{(k-1)}$
\EndFor
\EndIf
\State\Return $F_\alpha(\vx^{(K)})$
\end{algorithmic}
\end{algorithm}

{We remark here that while the absolute constant in the iteration count for the $\alpha > 1$ may appear large, we expect the actual constant to be much smaller in practice. This is because we have made no effort to optimize the constants, and have instead focused on reducing the dependence on problem parameters $\epsilon, m, n,$ and $\rho.$ We also note that it is possible to improve the empirical performance of the algorithm by using the autoscale idea from~\cite{d-allen2014using}. The autoscale idea speeds up the convergence in the initial iterations in which all of the constraints are loose by permitting the coordinates of $\vx^{(k)}$ to increase at a faster rate. In particular, if, for some $j,$ it holds $(\mA F_\alpha(\vx^{(k)}))_i \leq 1-\epsilon$ for all $i$ with $A_{ij} \neq 0,$ then $\hat{x}_j = (F_\alpha(\vx^{(k)}))_j$ is scaled by $\frac{1-\epsilon}{\max_{i:A_{ij}\neq 0}(\mA F_\alpha(\vx^{(k)}))_i}.$  Implementing autoscale has no effect on the analysis.}

We start by characterizing the ``local smoothness'' of $f_r$ which will be crucial for the analysis.%

\subsection{Local Smoothness and Feasibility}

The following lemma characterizes the step sizes that are guaranteed to decrease the function value. Since the algorithm makes multiplicative updates for $\alpha\neq 1$, we will require that $\vx > 0$, which will hold throughout, due to the particular initialization and the choice of the steps. {The proof of Lemma~\ref{lemma:smoothness} is provided in the appendix. }

\begin{lemma}\label{lemma:smoothness} Suppose that $\alpha \neq 1$ and $\vx > 0.$  
If $\Gamma = \diag(\vgamma)$, and $\gamma_j = -\frac{c_j}{4}\cdot \frac{\beta(1-\alpha)}{1+\alpha\beta}\tgradj_r(\vx)$ for $c_j \in [0, 1],$ then:
$$
f_r(\vx + \Gamma\vx) - f_r(\vx) \leq
 \sum_{j=1}^n \big(1-\frac{c_j}{2}\big)\gamma_j  x_j \nabla_j f_r(\vx).
$$
If $\alpha = 1$ and $\Delta\vx \geq 0$ is such that $\Delta x_j =-\frac{c_j\beta}{4(1+\beta)}\tgradj_r(\vx)$ for $c_j \in[0, 1],$ then:
$$
f_r(\vx + \Delta \vx) - f_r(\vx) \leq
 \sum_{j=1}^n \big(1- \frac{c_j}{2}\big)\Delta x_j \nabla_j f_r(\vx).
$$
\end{lemma}

Lemma~\ref{lemma:smoothness} also allows us to guarantee that the algorithm always maintains feasible solutions, as stated in the following proposition.

\begin{proposition}\label{prop:feasibility}
Solution $\vx^{(k)}$ computed by \textsc{FairPacking} at any iteration $k \geq 0$ is \eqref{eq:packing-transformed}-feasible. {Equivalently, $F_\alpha(\vx^{(k)})$ is \eqref{eq:alpha-fair-packing}-feasible.}
\end{proposition}
\begin{proof}
By the initialization and steps of \textsc{FairPacking}, $\vx^{(k)}\in S_\alpha,$ $\forall k.$ It remains to show that it must be $\mA F_\alpha(\vx^{(k)}) \leq \ones,$ $\forall k$. Observe that $\mA \vx^{(0)}\leq (1-\epsilon)\ones.$ Suppose that in some iteration $k$, $\exists i$ such that $(\mA F_\alpha(\vx^{(k)}))_i \geq 1-\epsilon/8.$ Fix one such $i$ and let $k$ be first such iteration. We provide the proof for the case when $\alpha < 1$. The cases $\alpha=1$ and $\alpha > 1$ follow by similar arguments.

Assume that $\alpha < 1.$ Then for all $j$ such that $A_{ij}(x_j^{(k)})^\frac{1}{1-\alpha} \geq \frac{1}{4n}$ (there must exist at least one such $j,$ as $(\mA F_\alpha(\vx))_i \geq 1-\frac{\epsilon}{8}\geq \frac{7}{8}$), we have  $(x_j^{(k)})^\frac{1}{1-a} \geq \frac{1}{4n\rho}$ and 
$
\nabla_j f_r(\vx^{(k)}) \geq \frac{1}{1-\alpha}\big(-1 + \big(\frac{1}{4n\rho}\big)^{\alpha}(1+\frac{\epsilon}{4})^\frac{1}{\beta}\big)>\frac{1}{1-\alpha}.
$
Hence, using Lemma~\ref{lemma:smoothness}, {all $x_j$ such that $A_{ij}(x_j^{(k)})^\frac{1}{1-\alpha} \geq \frac{1}{4n}$ must decrease, which implies} $(\mA (\vx^{(k+1)})^{\frac{1}{1-\alpha}})_i\leq (\mA (\vx^{(k)})^{\frac{1}{1-\alpha}})_i$. {On the other hand, by Lemma 3.1, the maximum increase of any coordinate in any iteration is by a factor $1 + \frac{\beta(1-\alpha)}{4(1+\alpha\beta)}.$} {Thus,} the maximum increase in $(\mA (\vx^{(k)})^{\frac{1}{1-\alpha}})_i$ in any iteration is by a factor {at most} ${(1 + \frac{\beta(1-\alpha)}{4(1+\alpha\beta)})^{1-\alpha}\leq e^{\frac{\beta}{4}}\leq }1+ \frac{\epsilon}{8},$  {and} it follows that it must be $\mA (\vx^{(k)})^{\frac{1}{1-\alpha}} \leq \ones,$  $\forall k.$ {The equivalence of feasibility of $F_\alpha(\vx^{(k)})$ was already discussed in Section~\ref{sec:fair-pc}.}
\end{proof}

\subsection{Main Theorem}\label{sec:main-thm}

Our main results are summarized in the following theorem. The theorem is proved through Theorems~\ref{thm:a=01convergence}, \ref{thm:a=1-convergence}, and~\ref{thm:a>1-convergence} in Section~\ref{sec:proofs}.

\begin{theorem}~\label{thm:main-thm}
Given $\mA$, $\alpha \geq 0$, and $\eps \in \big(0, \min\big\{\frac{1}{2}, \frac{1}{10|\alpha-1|}\big\}],$ let $\vx_\alpha^{(K)} = F_\alpha(\vx^{(K)})$ be the solution produced by \textsc{FairPacking} and let $\vx^*_\alpha$ be the optimal solution to~\eqref{eq:alpha-fair-packing}. Then $\vx_\alpha^{(K)}$ is \eqref{eq:alpha-fair-packing}-feasible and $f_\alpha(\vx^*_\alpha)-f_\alpha(\vx_\alpha^{(K)}) = O(\eps_f),$ where:
\begin{equation*}
\eps_f = \begin{cases}
\eps n, \quad &\text{ if } \alpha = 1,\\
\eps(1-\alpha)f_\alpha(\vx^*_\alpha), \quad &\text{ if } \alpha \neq 1.
\end{cases}
\end{equation*}
The total number of iterations taken by the algorithm 
is:
\begin{equation*}
K = \begin{cases}
O\big(\frac{\log(n\rho)\log(mn\rho/\eps)}{(1-\alpha)^3\eps^2}\big), \quad &\text{ if } \alpha \in [0, 1),\\
O\big(\frac{\log^3(\rho mn/\eps)}{\eps^2}\big), \quad &\text{ if } \alpha = 1, \\
O\big(\max\big\{\frac{\alpha^3\log(1/\epsilon)\log(mn\rho/\epsilon)}{\epsilon}, \frac{\log(\frac{1}{\epsilon(\alpha-1)})\log(mn\rho/\epsilon)}{\epsilon(\alpha-1)}\big\}\big)
, \quad &\text{ if } \alpha > 1. 
\end{cases}
\end{equation*}
\end{theorem}

\subsection{Approximate Duality Gap}\label{sec:adgt}

The proof relies on the construction of an approximate duality gap, in the framework of~\cite{thegaptechnique}. The idea is to construct an estimate of the optimality gap for the running solution. Namely, we want to show that an estimate of the true optimality gap $-f_\alpha(\vx_\alpha^{(k)})+f_\alpha(\vx^*_\alpha)$ decreases as the function of the iteration count $k,$ where $\vx_\alpha^{(k)} = F_\alpha(\vx^{(k)})$ (recall that, by Proposition~\ref{prop:feasibility}, $\vx_\alpha^{(k)}$ is \eqref{eq:alpha-fair-packing}-feasible). By construction of $f_r,$ we have that $f_r(\vx^{(k)})\geq -f_\alpha(\vx_\alpha^{(k)}),$ hence it is an upper bound on $-f_\alpha(\vx_\alpha^{(k)}).$ In the analysis, we will use $U_k = f_r(\vx^{(k+1)})$ as the upper bound. The lower bound $L_k$ needs to satisfy $L_k \leq - f_\alpha(\vx^*_\alpha)$. The approximate optimality (or duality, see~\cite{thegaptechnique}) gap at iteration $k$ is defined as $G_k = U_k - L_k.$

The goal is to show that $G_k$ decreases at rate $1/H_k$; namely, the idea is to show that $H_k G_k \leq H_{k-1} G_{k-1} + E_k$ for an increasing sequence of positive numbers $H_k$ and some ``sufficiently small'' error $E_k$. This argument is equivalent to stating that 
$$
-f_\alpha(\vx_\alpha^{(k+1)}) + f_\alpha(\vx_\alpha^*) \leq U_k - L_k =G_k \leq \frac{H_0}{H_k}G_0 + \frac{\sum_{\ell = 1}^k E_\ell}{H_k},
$$ 
which gives the standard form of convergence for first-order methods. %

For $\alpha > 1,$ it is unclear how to initialize the algorithm to guarantee a sufficiently small initial gap (and the right change in the gap in general). Instead, we will only require that the gap argument is valid on a subsequence of the iterates. We will argue that in the remaining iterations, $f_r$ must decrease by a  large multiplicative factor, so that either way we approach a $(1+\epsilon)-$approximate solution  at the target rate.

\paragraph*{Local Smoothness and the Upper Bound}

As already mentioned, our upper bound of choice will be $U_k = f_r(\vx^{(k+1)})$. The reason that the upper bound ``looks one step ahead'' is that it will hold a sufficiently lower value than $f_r(\vx^{(k)})$ (it will always decrease, due to Lemma~\ref{lemma:smoothness}) to compensate for any decrease in the lower bound $L_k$.%

\paragraph*{Lower Bound} 
Let $\{h_{\ell}\}_{\ell=0}^k$ be a sequence of positive real numbers such that $H_k = \sum_{\ell=0}^k h_\ell$. 
The simplest lower bound is just a consequence of convexity of $f_r$ and the fact that it closely approximates $-f_\alpha$ (due to Proposition~\ref{prop:regularization}): 
\begin{equation}\label{eq:trivial-lb}
\begin{aligned}
-f_\alpha(\vx^*_\alpha) \geq f_r(\vx_r^*) - 2\epsilon_f &\geq \frac{\sum_{\ell=0}^k h_\ell \left(f_r(\vx^{(\ell)}) + \innp{\nabla f_r(\vx^{(\ell)}), \vx_r^* - \vx^{(\ell)}}\right) }{H_k}.
\end{aligned}
\end{equation}
Even though simple, we will show that this lower bound can be used for the %
analysis of the $\alpha =1$ case. However, this lower bound is not useful in the case of $\alpha \neq 1.$ The reason comes as a consequence of the ``gradient-descent-type'' decrease from Lemma~\ref{lemma:smoothness}. While for $\alpha =1,$ the decrease can be expressed solely as the function of the gradient $\nabla f(\vx^{(k)})$ (and global problem parameters), when $\alpha \neq 1,$ the decrease is also a function of the current solution $\vx^{(k)}$. 
This means that %
we would need to be able to relate $\sum_{j=1}^n {(x_j^{(k)} - x_j^*)^2}/{x_j^{(k)}}$ to the value of $f_r(\vx^*),$ which is not even clear to be possible (see the convergence argument from Section~\ref{sec:alpha=1} for more information).

However, for $\alpha < 1,$ it is possible to obtain a useful lower bound from~\eqref{eq:trivial-lb} after performing gradient truncation and regularization, similar as in our note on packing and covering LP~\cite{LP-jelena-lorenzo}. Denote $\vx^* = \vx^*_r = \argmin_{\vu}f_r(\vu)$. We have:
\begin{align*}
f_r(\vx^*) \geq & \dfrac{\sum_{\ell=0}^k h_\ell (f_r(\vx^{(\ell)}) - \innp{\nabla f_r(\vx^{(\ell)}), \vx^{(\ell)}}) + \sum_{\ell = 0}^k h_\ell \innp{\frac{\tgrad_r(\vx^{(\ell)})}{1-\alpha},\vx^* }}{H_k}
\end{align*}
 Let $\phi: \mathbb{R}^n_+ \rightarrow \mathbb{R}$ be a convex function (that will be specified later). Adding and subtracting $\frac{\phi(\vx^*)}{1-\alpha}$ to the right-hand side of the last inequality, and then replacing $\vx^*$ with the minimizer of $\sum_{\ell = 0}^k h_\ell \innp{{\tgrad_r(\vx^{(\ell)})},\vu } + \phi(\vu)$ over $\vu \geq 0,$ we get
\begin{align*}
f_r(\vx^*) \geq L_k^{\alpha<1} + 2\epsilon_f \defeq & \frac{\sum_{\ell=0}^k h_\ell (f_r(\vx^{(\ell)}) - \innp{\nabla f_r(\vx^{(\ell)}), \vx^{(\ell)}}) -\frac{1}{1-\alpha}\phi(\vx^*)}{H_k}\\
&
+ \frac{\min_{\vu \geq \zeros}\big\{\sum_{\ell = 0}^k h_\ell \innp{{\tgrad_r(\vx^{(\ell)})},\vu } + \phi(\vu)\big\}}{(1-\alpha)H_k}.%
\end{align*}
Note that the same lower bound cannot be derived for $\alpha \geq 1$. The reason is that we cannot perform gradient truncation, as for $\alpha> 1$ (resp.~$\alpha = 1$), $\innp{\nabla f_r(\vx), \vx^*}\geq \frac{1}{1-\alpha}\innp{\tgrad_r(\vx), \vx^*}$ (resp.~$\innp{\nabla f_r(\vx), \vx^*}\geq \innp{\tgrad_r(\vx), \vx^*}$) \emph{does not hold}. 

For $\alpha > 1,$ we make use of the Lagrangian dual of \eqref{eq:packing-transformed} $g:\mathbb{R}_+^m \to \mathbb{R}$ given by
\begin{equation}
g(\vy) =
-\innp{\ones, \vy} - \frac{\alpha}{1-\alpha}\sum_{j=1}^n (\mA^T \vy)_j^{-\frac{1-\alpha}{\alpha}}.
\end{equation}
Finally, we note that it is not clear how to use the Lagrangian dual in the case of $\alpha \leq 1$. When $\alpha < 1,$ the terms $- \frac{1}{1-\alpha}(\mA^T\vy)_j^{-\frac{1-\alpha}{\alpha}}$ approach $-\infty$ as $(\mA^T\vy)_j$ approaches zero. A similar argument can be made for $\alpha=1,$ in which case the Lagrangian dual is $g(\vy) = - \innp{\ones, \vy} + n +\sum_{j=1}^n \log(\mA^T\vy)_j.$ In~\cite{marasevic2015fast}, this was handled by ensuring that $(\mA^T\vy)_j$ never becomes ``too small,'' which requires step sizes that are smaller by a factor $\epsilon$ and generally leads to much slower convergence.
%
%
%
\section{Proof of the Main Theorem}\label{sec:proofs}

In this section, we provide the complete proof of the Main theorem (Theorem~\ref{thm:main-thm}). The proof is provided by proving three separate theorems (Theorem~\ref{thm:a=01convergence}, Theorem~\ref{thm:a=1-convergence}, and Theorem~\ref{thm:a>1-convergence}), each dealing separately with the cases $\alpha \in [0, 1),$ $\alpha=1,$ and $\alpha > 1,$ respectively.

\subsection{Convergence Analysis for $\alpha \in [0, 1)$}\label{sec:alpha<1}
{Recall that in this setting, the algorithm makes updates of the following form:}
\begin{equation*}
    \begin{gathered}
        {\vx^{(k)} = (\ones + \vz^{(k-1)})^{-1/\beta'},}\\
       {\vz^{(k)} = \vz^{(k-1)} + \epsilon h \tgrad_r(\vx^{(k)}).}
    \end{gathered}
\end{equation*}

To analyze the convergence of \textsc{FairPacking}, we need to specify $\phi(\cdot)$ from the lower bound $L_k^{\alpha<1}$ introduced in Section~\ref{se:algo}. To simplify the notation, in the rest of the section, we use $L_k$ to denote $L_k^{\alpha<1}.$ We define $\phi$ in two steps, as follows:
\begin{equation}\label{eq:a<1-phi-def}
\begin{gathered}
\phi(\vx) \defeq \psi(\vx) - \big\langle{\nabla \psi(\vx^{(0)}) + h_0 \tgrad_r(\vx^{(0)}), \vx}\big\rangle,\\
\psi(\vx) \defeq \frac{1}{\epsilon}\Big( \innp{\ones, \vx} - \frac{\langle \ones, \vx^{1-\beta'} \rangle}{1-\beta'} \Big)%
\end{gathered}
\end{equation}
where $\beta' = \frac{(1-\alpha)\epsilon/4}{\log(n\rho/(1-\epsilon))}$.  
This particular choice of $\phi$ is made for the following reasons. First, $\epsilon\psi(\vx)$ closely approximates $\innp{\ones, \vx}$ (up to an $\epsilon$ multiplicative factor, unless $\innp{\ones, \vx}$ is negligible). This will ensure that $\frac{1}{1-\alpha}\phi(\vx^*)$ is within $O(1-\alpha)f_\alpha(\vx^*_\alpha),$ which will allow us to bound the initial gap by $O(1-\alpha)f_\alpha(\vx^*).$ 

To understand the role of $-\innp{\nabla \psi(\vx^{(0)}) + h_0 \tgrad_r(\vx), \vx},$ notice that the steps of \textsc{FairPacking} are defined as 
$%
\vx^{(k+1)} = \argmin_{\vu \geq \zeros} \big\{\sum_{\ell = 0}^k h_\ell \innp{{\tgrad_r(\vx^{(\ell)})},\vu } + \phi(\vu)\big\}.
$ %
The role of the term $-\innp{\nabla \psi(\vx^{(0)}) + h_0 \tgrad_r(\vx^{(0)}), \vx}$ is to ensure that $\vx^{(1)}=\vx^{(0)}$,
which will allow us to properly initialize the gap. Finally, the scaling factor $\frac{1}{\epsilon}$ ensures that $\vz^{(k)} \leq 1 + \epsilon/2$ (see the proof of Lemma~\ref{lemma:a>1-gap-decrease}), which will allow us to argue that the steps satisfy the assumptions of Lemma~\ref{lemma:smoothness}. We also need %
to guarantee that:
\begin{equation}\label{eq:a<1-zk-def}
\vz^{(k)} \defeq \epsilon\bigg(\sum_{\ell = 1}^k h_\ell\tgrad_r(\vx^{(\ell)})-\nabla \psi(\vx^{(0)})\bigg)
\end{equation}
is bounded below by $-O(\epsilon)$ to ensure that the upper bound can compensate for any decrease in the lower bound. Some of these statements are formalized  below.
\begin{proposition}\label{prop:a<1-initial-conditions}
Let $\vz^{(k)},$ $\psi(\cdot),$ and $\phi(\cdot)$ be defined as in Equations~\eqref{eq:a<1-phi-def}, \eqref{eq:a<1-zk-def}. Define: $\psih(\vz^{(k)}) = - \frac{\beta'/\epsilon}{1-\beta'}\sum_{j=1}^n (1 + z_j^{(k)})^{-\frac{1-\beta'}{\beta'}}.$ Then:
\begin{enumerate}
\item $\psih(\vz^{(k)}) = \min_{\vu \geq \zeros} \big\{\sum_{\ell = 0}^k h_\ell \innp{{\tgrad_r(\vx^{(\ell)})},\vu } + \phi(\vu)\big\};$
\item $\vx^{(1)} = \argmin_{\vu \geq \zeros}  \{h_0 \innp{{\tgrad_r(\vx^{(0)})},\vu } + \phi(\vu)\}= \vx^{(0)}$ and $\vx^{(k+1)} =  \epsilon\nabla\psih(\vz^{(k)})$.
\item $\epsilon/4\leq \vz^{(0)}\leq \epsilon/2.$
\end{enumerate}
\end{proposition}
\begin{proof}
The first part follows directly from the definitions of $\vz^{(k)}$ and $\phi,$ using the first-order optimality condition to solve the minimization problem that defines $\psih$.

For the second part, by the definition of $\phi$ and the first-order optimality condition:
$$
\argmin_{\vu \geq \zeros}  \big\{h_0 \innp{{\tgrad_r(\vx^{(0)})},\vu } + \phi(\vu)\big\} = \argmin_{\vu \geq \zeros}\big\{\psi(\vu) - \innp{\nabla \psi(\vx^{(0)}), \vu}\big\} = \vx^{(0)}.
$$
Similarly, for $\vx^{(k+1)},$ we have $\vx^{(k+1)} = \argmin_{\vu\ge\zeros}\{\innp{\vz^{(k)}, \vu} + \epsilon \psi(\vu)\}$. It is not hard to verify  that $x_j^{(k+1)} = (1+z_j^{(k)})^{-1/\beta'}=\epsilon\nabla_j\psih(\vz^{(k)}).$ 
For the last part, recall that $x_j^{(0)} = \big(\frac{1-\epsilon}{n\rho}\big)^{\frac{1}{1-\alpha}}$ and observe that $\nabla_j\psi(\vx) = \frac{1}{\epsilon}(1-{x_j}^{-\beta'})$. Hence:
$$
z_j^{(0)} = \big(\frac{n\rho}{1-\epsilon}\big)^{\frac{\beta'}{1-\alpha}}-1 = \big(\frac{n\rho}{1-\epsilon}\big)^{\frac{\epsilon/4}{\log(n\rho/(1-\epsilon))}}-1 = \exp(\epsilon/4)-1.
$$
The rest of the proof follows by approximating $\exp(\epsilon/4)$.
\end{proof}
Using Proposition~\ref{prop:a<1-initial-conditions}, we can now bound the initial gap, as follows.
\begin{proposition}\label{prop:a<1-init-gap}
Let $h_0 = H_0 = 1$. Then $H_0G_0 - 2\epsilon (1-\alpha) f_\alpha(\vx^*_\alpha)\leq 2 f_\alpha(\vx^*_\alpha).$
\end{proposition}
\begin{proof}
From Proposition~\ref{prop:a<1-initial-conditions}, $U_1 = f_r(\vx^{(1)})=f_r(\vx^{(0)})$ and thus:
$
H_0G_0 = \frac{-\psih(\vz^{(0)}) + \phi(\vx^*)}{1-\alpha} + 2\epsilon(1-\alpha)f_\alpha(\vx_{\alpha}^*).
$ 
The rest of the proof follows by bounding $\psih(\vz^{(0)})$ and $\phi(\vx_r^*)$. For the former, it is not hard to verify that $\sum_{j=1}^n\frac{(x_j^{(0)})^{1-\beta'}}{1-\beta'}\leq (1+\epsilon/2)\innp{\ones, \vx^{(0)}}$. Hence, as $x_j^{(1)}=x_j^{(0)} = (z_j^{(0)})^{-1/\beta'},$ we have:
\begin{align*}
-\psih(\vz^{(0)})\leq \frac{\beta'(1+\epsilon/2)}{\epsilon}\innp{\ones, \vx^{(0)}}\leq \frac{1}{2}(1-\alpha)\innp{\ones, \vx^{(0)}} \leq \frac{1}{2}(1-\alpha)^2f_\alpha(\vx^*_\alpha).
\end{align*}

For the latter, observe first that as $\vx^* \leq \ones$ (by feasibility, Proposition~\ref{prop:feasibility}), it must be $\psi(\vx^*)\leq 0.$ Hence, we can finally bound $\phi(\vx^*)$ as:
\begin{align*}
\phi(\vx^*) &\leq - \innp{\nabla\psi(\vx^{(0)}) + h_0 \tgrad_r(\vx^{(0)}), \vx^*}\leq -\innp{(-1/2-1)\ones, \vx^*}\\
&\leq \frac{3}{2}\innp{\ones, \vx^*}\leq \frac{3}{2}(1-\alpha)f_\alpha(\vx_\alpha^*),
\end{align*}
as $\tgrad_r(\vx) \geq -\ones,$ $\forall \vx$ and $\nabla\psi(\vx^{(0)})=\vz^{(0)}/\epsilon\geq -(1/2)\ones$ (due to Proposition~\ref{prop:a<1-initial-conditions}).
\end{proof}
The crucial part of the convergence analysis is to show that for some choice of step sizes $h_k,$ $H_kG_k \leq H_{k-1}G_{k-1} + 2h_k \epsilon_f.$ Note that to make the algorithm as fast as possible (since its convergence rate is proportional to $H_k$), we would like to set $h_k$'s as large as possible. However, enforcing  the condition $H_kG_k \leq H_{k-1}G_{k-1}+ 2h_k \epsilon_f$ will set an upper bound on the choice of $h_k.$ We have the following lemma.
\begin{lemma}\label{lemma:a<1-gap-decrease} 
If $G_{k-1}-2\epsilon_f\leq 2f_\alpha(\vx^*_\alpha)$ and $ h_k \leq \frac{(1-\alpha)\beta\beta'}{16 \epsilon(1+\beta)} = \theta (\frac{(1-\alpha)^2\epsilon}{\log(n\rho)\log(mn\rho/\epsilon)}),$ then $H_kG_k \leq H_{k-1}G_{k-1} + 2h_k\epsilon_f,$ $\forall k \geq 1.$
\end{lemma}
\begin{proof}
The role of the assumption $G_{k-1}-2\epsilon_f\leq 2(1-\alpha)f_\alpha(\vx^*_\alpha)$ is to guarantee that $\vz^{(k-1)}\geq -(\epsilon/2)\ones.$ Namely, if $z_j^{(k-1)}< -\epsilon/2,$ for any $j$, $\psi^*(\vz^{(k-1)})$ blows up, making the gap $G_{k-1}$ much larger than $3(1-\alpha)f_\alpha(\vx^*_\alpha).$ This is not hard to argue (see also a similar argument in~\cite{LP-jelena-lorenzo}) and hence we omit the details and assume from now on that $\vz^{(k-1)}\geq -(\epsilon/2)\ones.$ Note that this assumption holds initially due to Proposition~\ref{prop:a<1-initial-conditions}. Observe that as $\epsilon < 1/H_k$ and $\tgrad_r(\vx^{(\ell)})\leq \ones,$ $\forall \ell,$ we also have 
$$
z_j^{(k)}{= \eps\Big(\sum_{\ell=1}^k h_{\ell}\tgradj_r(\vx^{(k)}) - \nabla_j \psi(\vx^{(0)})\Big) \leq 1 - \epsilon \nabla_j \psi(\vx^{(0)})} \leq 1+\epsilon/2,
$$
{where we have used $\nabla_j \psi(\vx^{(0)}) = z_j^{(0)}/\eps \geq - \frac{1}{2}$ (due to Proposition~\ref{prop:a<1-initial-conditions}, Part 3).}
To be able to apply Lemma~\ref{lemma:smoothness}, we need to ensure that 
$$
|x_j^{(k+1)} - x_j^{(k)}| \leq c_j \frac{\beta}{4(1+\beta)}|\tgradj_r(\vx^{(k)})|x_j^{(k)},
$$
for all $j$ and for $c_j \in (0, 1].$ Recalling the definition of $\vx^{(k+1)}$, $x_j^{(k+1)} = \nabla_j \psih(\vz^{(k)})= (1 + z_j^{(k)})^{-1/\beta'}.$ As $z_j^{(k)} = z_j^{(k-1)} + \epsilon h_k\tgradj_r(\vx^{(k)})$, we have:
$$
x_j^{(k+1)} = (1 + z_j^{(k-1)} + h_k \epsilon \tgradj_r(\vx^{(k)}))^{-1/\beta}=x_j^{(k)}\bigg(1 + \frac{\epsilon h_k\tgradj_r(\vx^{(k)})}{1+z_j^{(k-1)}}\bigg)^{-1/\beta'}.
$$
Suppose first that $\tgradj_r(\vx^{(k)})\leq 0.$ Then $\frac{\tgradj_r(\vx^{(k)})}{1-\epsilon/2}\leq \frac{\tgradj_r(\vx^{(k)})}{1+z_j^{(k-1)}}\leq \frac{\tgradj_r(\vx^{(k)})}{2+\epsilon/2}.$ As $\frac{\epsilon h_k}{\beta'} \leq (1-\epsilon/2)\frac{(1-\alpha)\beta}{8(1+\beta)}$ and $\tgradj_r(\vx^{(k)})\geq -1$ we have:
\begin{equation*}
\begin{aligned}
1-\frac{1-\epsilon/2}{2+\epsilon/2}\cdot\frac{(1-\alpha)\beta}{8(1+\beta)}\tgradj_r(\vx^{(k)}) &\leq \Big(1 + \frac{\epsilon h_k\tgradj_r(\vx^{(k)})}{1+z_j^{(k-1)}}\Big)^{-1/\beta'}\\
&\leq 1 - \frac{(1-\alpha)\beta}{4(1+\beta)}\tgradj_r(\vx^{(k)}). 
\end{aligned}
\end{equation*}
Similarly, when $\tgradj_r(\vx^{(k)})> 0,$  $\frac{\tgradj_r(\vx^{(k)})}{2+\epsilon/2}\leq \frac{\tgradj_r(\vx^{(k)})}{1+z_j^{(k-1)}}\leq \frac{\tgradj_r(\vx^{(k)})}{1-\epsilon/2}.$ As $\frac{\epsilon h_k}{\beta'} \leq (1-\epsilon/2)\frac{(1-\alpha)\beta}{8(1+\beta)}$ and $\tgradj_r(\vx^{(k)})\leq 1$ we have:
\begin{equation*}
\begin{aligned}
1 - \frac{(1-\alpha)\beta}{4(1+\beta)}\tgradj_r(\vx^{(k)}) 
&\leq \Big(1 + \frac{\epsilon h_k\tgradj_r(\vx^{(k)})}{1+z_j^{(k-1)}}\Big)^{-1/\beta'}\\
&\leq 1-\frac{1-\epsilon/2}{2+\epsilon/2}\cdot\frac{(1-\alpha)\beta}{8(1+\beta)}\tgradj_r(\vx^{(k)}). 
\end{aligned}
\end{equation*}
Either way, Lemma~\ref{lemma:smoothness} can be applied with $c_j \geq \frac{1-\epsilon/2}{2(2+\epsilon/2)}\geq \frac{1}{10}$, and we have:
\begin{equation}\label{eq:a<1-ub-change}
\begin{aligned}
H_kU_k - H_{k-1}U_{k-1} &{= H_k (U_k - U_{k-1}) + h_k U_{k-1}}\\
&{= H_k(f_r(\vx^{(k+1)}) - f_r(\vx^{(k)})) + h_k f_r(\vx^{(k)})}\\
&\leq h_k f_r(\vx^{(k)})-\frac{H_k\beta}{50(1+\alpha\beta)}\sum_{j=1}^n x_j^{(k)}\nabla_j f_r(\vx^{(k)})\tgradj_r(\vx^{(k)}).
\end{aligned}
\end{equation}
On the other hand, the change in the lower bound is:
\begin{equation}\label{eq:a<1-lb-change-1}
\begin{aligned}
& H_kL_k - H_{k-1}L_{k-1}\\
&\hspace{.2cm}= h_k \big(f_r(\vx^{(k)})-\innp{\nabla f_r(\vx^{(k)}), \vx^{(k)}}\big) +\frac{1}{1-\alpha}\left(\psih(\vz^{(k)}) - \psih(\vz^{(k-1)})\right) + 2h_k \epsilon_f.
\end{aligned}
\end{equation}
Using Taylor's Theorem:  
\begin{equation}\label{eq:a<1-taylor}
\begin{aligned}
&\psih(\vz^{(k)}) - \psih(\vz^{(k-1)}) \\
&\hspace{.2cm}= \innp{\nabla\psih(\vz^{(k-1)}), \vz^{(k)}-\vz^{(k-1)}} + \frac{1}{2}\innp{\nabla^2\psih(\vzh)(\vz^{(k)}-\vz^{(k-1)}), \vz^{(k)}-\vz^{(k-1)}},
\end{aligned}
\end{equation}
where $\vzh = \vz^{(k-1)}+t(\vz^{(k)}-\vz^{(k-1)}),$ for some $t\in [0, 1].$ 
Recall that $\epsilon\nabla_j \psih(\vz^{(k-1)}) = (1+z_j^{(k-1)})^{-1/\beta'} = x_j^{(k)}$ and $\vz^{(k)}-\vz^{(k-1)} = \epsilon h_k \tgrad_r(\vz^{(k)}).$ Observe that $\nabla^2_{jj}\psih(\vz) = - \frac{1}{\epsilon\beta'}(1+z_j)^{-{(1+\beta')}/{\beta'}},$ $\nabla^2_{jk}\psih(\vz) = 0,$ for $j\neq k.$ As $z_j^{(k-1)} \geq -\epsilon/2$ and $z_j^{(k)}=z_j^{(k-1)} + \epsilon h_k \tgradj_r(\vx^{(k)})\geq 1-\epsilon h_k,$ we have that: 
$$(1+z_j^{(k)})^{-\frac{1-\beta'}{\beta'}}\leq (1-\epsilon h_k)^{-\frac{1-\beta'}{\beta'}}(1+z_j^{(k-1)})^{-\frac{1-\beta'}{\beta'}}< (1+\epsilon/2)(1+z_j^{(k-1)})^{-\frac{1-\beta'}{\beta'}},$$
as $\frac{\epsilon h_k}{\beta'} \leq (1-\epsilon/2)\frac{(1-\alpha)\beta}{8(1+\beta)}$. Further, as $x_j^{(k)} = (1+z_j^{(k-1)})^{-1/\beta'}$ and $z_j^{(k)}\geq 1-\epsilon/2,$ we have that $(1+z_j^{(k-1)})^{-\frac{1-\beta'}{\beta'}}\leq x_j^{(k)}/(1-\epsilon/2).$  Hence,~\eqref{eq:a<1-taylor} implies:
\begin{align}\label{eq:a<1-psih-change}
\psih(\vz^{(k)})-\psih(\vz^{(k-1)})\geq h_k \innp{\tgrad_r(\vx^{(k)}), \vx^{(k)}} - \frac{3(\epsilon h_k)^2}{2\beta'}\sum_{j=1}^n x_j^{(k)}(\tgradj_r(\vx^{(k)}))^2.
\end{align}

Using~\eqref{eq:a<1-ub-change},~\eqref{eq:a<1-lb-change-1}, and~\eqref{eq:a<1-psih-change}, to complete the proof, it suffices to show that, $\forall j,$
\begin{align*}
\xi_j \defeq & h_k \Big(\nabla_j f_r(\vx^{(k)})-\frac{1}{1-\alpha}\tgradj_r(\vx^{(k)})\Big) + \frac{3(\epsilon h_k)^2}{2\beta'(1-\alpha)}(\tgradj_r(\vx^{(k)}))^2\\
&- \frac{H_k\beta}{50(1+\alpha\beta_k)}\nabla_j f_r(\vx^{(k)})\tgradj_r(\vx^{(k)})
\leq  0.
\end{align*}
Consider the following two cases for $\tgradj_r(\vx^{(k)}):$

\noindent\textbf{Case 1:} $\tgradj_r(\vx^{(k)})< 1.$ Then $\tgradj_r(\vx^{(k)}) = (1-\alpha)\nabla_j f_r(\vx^{(k)}),$ and we have:
\begin{align*}
\xi_j = \frac{(\tgradj_r(\vx^{(k)}))^2}{1-\alpha}\Big(\frac{3(\epsilon h_k)^2}{2\beta'} - \frac{H_k\beta}{50(1+\alpha\beta)}\Big) < 0,
\end{align*}
as ${\epsilon h_k} \leq (1-\epsilon/2)\frac{{\beta'}(1-\alpha)\beta}{8(1+\beta)}.$

\noindent\textbf{Case 2:} $\tgradj_r(\vx^{(k)})=1$. Then $\nabla_j f_r(\vx^{(k)}) \geq \frac{1}{1-\alpha}\geq 1,$ and:
\begin{align*}
\xi_j = \frac{h_k}{1-\alpha}\Big(\frac{3\epsilon^2 h_k}{\beta'}-1\Big) + \nabla_j f_r(\vx^{(k)})\Big(h_k - \frac{H_k\beta}{50(1+\alpha\beta)}\Big)\leq 0,
\end{align*}
by the choice of $h_k.$
\end{proof}

We are now ready to bound the overall convergence of \textsc{FairPacking} for $\alpha<1.$

\begin{theorem}\label{thm:a=01convergence}
Let $h_0 = 1,$ $h_k = h = \frac{(1-\alpha)\beta\beta'}{16 \epsilon(1+\beta)}$ for $k \geq 1$. Then, after at most $K =\lceil\frac{2}{h(1-\alpha)\epsilon}\rceil = \theta(\frac{\log(n\rho)\log(mn\rho/\epsilon)}{(1-\alpha)^3\epsilon^2})$ iterations of \textsc{FairPacking},  we have that $\vx_a^{(K+1)} = F_\alpha(\vx^{(K+1)})^{\frac{1}{1-\alpha}} = (\vx^{(K+1)})^{1-\alpha}$ is \eqref{eq:alpha-fair-packing}-feasibile and:
$$
f_\alpha(\vx_a^{(K+1)}) - f_\alpha(\vx^*_\alpha)\geq - 3\epsilon(1-\alpha)f_\alpha(\vx^*_\alpha).
$$
\end{theorem}
\begin{proof}
Feasibility of $\vx_\alpha^{(K+1)}$ follows from Proposition~\ref{prop:feasibility}, as the steps of \textsc{FairPacking} satisfy the conditions of Lemma~\ref{lemma:smoothness}. 

Due to Proposition~\ref{prop:a<1-init-gap}, the assumptions of Lemma~\ref{lemma:a<1-gap-decrease} hold initially and hence they hold for all $k$ (as Lemma~\ref{lemma:a<1-gap-decrease} itself when applied to iteration $k$ implies that its assumptions hold at iteration $k+1$). Thus, we have: $G_K \leq \frac{H_0 G_0}{H_K} + \frac{\sum_{\ell=0}^K h_\ell 2\epsilon_f}{H_K} = \frac{H_0 G_0}{H_K} + 2\epsilon_f.$ As, from Proposition~\ref{prop:a<1-init-gap}, $H_0G_0 \leq 2 f_\alpha(\vx_\alpha^*)$ and $H_K =  Kh \geq \frac{2}{(1-\alpha)\epsilon}$, it follows that $G_k \leq 3\epsilon(1-\alpha)f_\alpha(\vx_\alpha^*).$ Finally, recalling that by construction, $-f_\alpha(\vx_\alpha^{(K+1)})+f_\alpha(\vx_\alpha^*)\leq G_K,$ the claimed statement follows.
\end{proof}
\subsection{Convergence Analysis for $\alpha = 1$}\label{sec:alpha=1}
{In this setting, the algorithm makes updates of the following form:}
\begin{equation*}
   { \vx^{(k)} = \vx^{(k-1)} - \frac{\beta}{4(1+\beta)}\tgrad_r(\vx^{(k-1)})}
\end{equation*}

Let us start by bounding the coordinates of the running solutions $\vx^{(k)},$ for each iteration $k$. This will allow us to bound the initial-gap-plus-error $H_0G_0 + \sum_{\ell = 1}^k E_i$ in the convergence analysis.
\begin{proposition}\label{prop:a=1-bounds-on-x}
In each iteration $k,$ $-\log(2\rho mC)\ones\leq \vx^{(k)}\leq \zeros$.
\end{proposition}
\begin{proof}
Using Proposition~\ref{prop:feasibility}, $\vx^{(k)}\leq \zeros$ follows immediately by $\min_{ij: A_{ij}\neq 0}A_{ij} = 1.$ 
Suppose that in some iteration $k,$ $x_j^{(k)} \leq -\log({2\rho mC})+\epsilon/4$. Then, by Proposition~\ref{prop:feasibility}, $AF_{\alpha}(\vx^{(k)}) \leq \ones,$ and it follows that:
\begin{align*}
\nabla_j f_r(\vx^{(k)}) &{= - 1 + C \sum_{i=1}^m A_{ij} e^{x_j}(AF_{\alpha}(\vx^{(k)}))_i^{1/\beta}} \\
&{\leq -1 + C \sum_{i=1}^m \rho e^{-\log(2\rho m C) + \eps/4}}\\
&\leq -1 + C\cdot \frac{1}{2\rho mC}\exp(\epsilon/4)\rho m \leq -\frac{1-\epsilon}{2},
\end{align*}
{where in the second line we have used that $A_{ij} \leq \rho,$ $\forall i, j.$}

Hence, $x_j^{(k)}$ must increase in iteration $k$. Since the maximum decrease in any coordinate and in any iteration is by less than $\epsilon/4,$ it follows that $\vx^{(k)}\geq - \log(2\rho mC)\ones.$ %
\end{proof}

Recall that $U_k = f_r(\vx^{(k+1)})$ and $L_k = \frac{\sum_{\ell=0}^k h_\ell \left(f(\vx^{(\ell)} + \innp{\nabla f_r(\vx^{(\ell)}), \vx^* - \vx^{(\ell)}})\right)}{H_k}-2\epsilon n$. Let us start by bounding the initial gap $G_0.$

\begin{proposition}\label{prop:a=1-init-gap}
$A_0 G_0 \leq E_0,$ where $E_0 = \frac{2(1+\beta)}{\beta} \cdot\frac{{h_0}^2}{H_0} \|\vx^* - \vx^{(0)}\|^2 + 2h_0\epsilon n$.
\end{proposition}
\begin{proof}
By the choice of the initial point $\vx^{(0)},$ it follows that $\nabla f_r(\vx^{(0)})\leq \zeros,$ and, thus $\tgrad_r(\vx^{(0)}) = \nabla f_r(\vx^{(0)})$. Using the Cauchy-Schwartz Inequality:
\begin{equation}\label{eq:a=1-cs}
H_0L_0 \geq h_0 f(\vx^{(0)}) - h_0 \|\nabla f_r(\vx^{(0)})\|\cdot\|\vx^* - \vx^{(0)}\| - 2H_0\epsilon n,
\end{equation}
while, from Lemma~\ref{lemma:smoothness},
\begin{equation}\label{eq:a=1-initi-uk}
H_0 U_0 \leq h_0 f(\vx^{(0)}) - H_0\frac{\beta}{8(1+\beta)}\|\nabla f_r(\vx^{(0)})\|^2.
\end{equation}
Combining \eqref{eq:a=1-cs} and \eqref{eq:a=1-initi-uk} with $-a^2 + 2ab \leq b^2,$ $\forall a, b,$ and as $H_0 = h_0,$ it follows that:
\begin{align*}
H_0 G_0 = H_0(U_0 - L_0) \leq \frac{2(1+\beta)}{\beta} \cdot\frac{{h_0}^2}{H_0}\|\vx^* - \vx^{(0)}\|^2 + 2h_0 \epsilon n.
\end{align*}
\end{proof}

The main part of the analysis is to show that for $k\geq 1,$ $H_k G_k - H_{k-1}G_{k-1} \leq E_k$, which, combined with Proposition~\ref{prop:a=1-init-gap} and the definition of the gap would imply $f(\vx^{(k+1)}) - f(\vx^*)\leq G_k \leq \frac{\sum_{i=0}^k E_i}{H_k},$ allowing us to bound the approximation error. %

\begin{lemma}\label{lemma:a=1-gap-decrease}
If, for $k\geq 1$, $\frac{h_k}{H_k} \leq \frac{\beta}{8(1+\beta)\log(2\rho mC)}$, then $H_k G_k - H_{k-1}G_{k-1} \leq E_k,$ where $E_k = \frac{2(1+\beta)}{\beta} \cdot\frac{{h_k}^2}{H_k}\|\vx^* - \vx^{(k)}\|^2$.
\end{lemma}
\begin{proof}
Applying Cauchy-Schwartz Inequality and Lemma~\ref{lemma:smoothness}:
\begin{align}\label{eq:a=1-lb-change}
H_k L_k - H_{k-1}L_{k-1} &\geq h_k f_r(\vx^{(k)})
 - h_k\sum_{j=1}^n |\nabla_j f_r(\vx^{(k)})|\cdot|x^*_j - x_j^{(k)}|,\\
H_k U_k - H_{k-1}U_{k-1} &\leq h_k f_r(\vx^{(k)}) - \frac{H_k\beta}{8(1+\beta)}\sum_{j=1}^n \tgradj_r(\vx^{(k)})\nabla_j f_r(\vx^{(k)}). \label{eq:a=1-ub-change}
\end{align}
Hence, combining \eqref{eq:a=1-lb-change} and \eqref{eq:a=1-ub-change}:
\begin{equation}\label{eq:a=1-gap-change}
\begin{aligned}
&H_k G_k - H_{k-1}G_{k-1}\\
&\hspace{.5cm}\leq \sum_{j=1}^n \Big(h_k |\nabla_j f_r(\vx^{(k)})|\cdot|x^*_j - x_j^{(k)}| - \frac{H_k\beta}{8(1+\beta)}\tgradj_r(\vx^{(k)})\nabla_j f_r(\vx^{(k)}) \Big).
\end{aligned}
\end{equation}
Let $e_j = h_k |\nabla_j f_r(\vx^{(k)})|\cdot|x^*_j - x_j^{(k)}| - \frac{H_k\beta}{8(1+\beta)}\tgradj_r(\vx^{(k)})\nabla_j f_r(\vx^{(k)})$ be the $j^{\mathrm{th}}$ term in the summation from the last equation, and consider the following two cases.

\noindent\textbf{Case 1:} $\nabla_j f_r(\vx^{(k)})\leq 1.$ Then $\tgradj_r(\vx^{(k)}) = \nabla_j f_r(\vx^{(k)})$ and using that $-a^2 + 2ab \leq b^2,$ $\forall a, b,$ it follows that:
\begin{equation}\label{eq:a=1-e-j<1}
e_j \leq \frac{2(1+\beta)}{\beta}\frac{{h_k}^2}{H_k} (x^*_j - x_j^{(k)})^2.
\end{equation}

\noindent\textbf{Case 2:} $\nabla_j f_r(\vx^{(k)}) > 1.$ Then $\tgradj_r(\vx^{(k)}) = 1$. By Proposition~\ref{prop:a=1-bounds-on-x}, $-\log(2\rho mC)\leq x_j^{(k)}\leq 0$ and similar bounds can be obtained for $x_j^*$ (see~\cite{marasevic2015fast}). It follows that:
\begin{align}\label{eq:a=1-e-j=1}
e_j \leq |\nabla_j f(\vx^{(k)})|\Big(h_k \log(2\rho mC)-\frac{H_k \beta}{8(1+\beta)}\Big)\leq 0,
\end{align}
as $\frac{h_k}{H_k} \leq \frac{\beta}{8(1+\beta)\log(2\rho mC)}$.

Combining \eqref{eq:a=1-gap-change}-\eqref{eq:a=1-e-j=1} completes the proof. %
\end{proof}

We are now ready to obtain the final convergence bound for $\alpha = 1:$

\begin{theorem}\label{thm:a=1-convergence}
If $k \geq 10 \frac{\log^2(2\rho mC)}{\epsilon\beta} = O\big(\frac{\log^3(\rho m n/\epsilon)}{\epsilon^2}\big),$ then $\vx_\alpha^{(k+1)} = \exp(\vx^{(k+1)})$ is \eqref{eq:alpha-fair-packing}-feasible and $f_\alpha(\vx_\alpha^{(k+1)}) - f_\alpha(\vx^*_\alpha) \geq - 3\epsilon n.$%
\end{theorem}
\begin{proof}
Combining Proposition~\ref{prop:a=1-init-gap} and Lemma~\ref{lemma:a=1-gap-decrease}, we have that if for $\ell \geq 1,$ $\frac{h_\ell}{H_\ell}\leq \lambda \defeq \frac{\beta}{8(1+\beta)\log(2\rho mC)}$, then $G_k \leq \frac{2(1+\beta)}{H_k\beta}\sum_{\ell=0}^k \frac{{h_{\ell}^2}}{H_\ell}\|\vx^* - \vx^{(\ell)}\|^2 + 2n\epsilon$. As discussed before, $\|\vx^* - \vx^{(\ell)}\|^2 \leq n\log^2(2\rho mC)$, and thus:
\begin{equation}\label{eq:a=1-f_r-bnd-1}
G_k \leq \frac{2(1+\beta)}{\beta} n\log^2(2\rho mC)\frac{1}{H_k}\sum_{\ell=0}^k \frac{{h_{\ell}^2}}{H_\ell} + 2n\epsilon.%
\end{equation}
As the sequence $\{h_\ell\}_{\ell = 1}^k$ does not affect the analysis, we can choose it arbitrarily, as long as $\frac{h_\ell}{H_\ell}\leq \lambda$ for $\ell \geq 1$. Let $h_0 = 1$ and  $\frac{h_{\ell}}{H_{\ell}} = \frac{\beta \epsilon }{8(1+\beta)\log^2(2\rho mC)}<\lambda$ for $\ell \geq 1$. Then:
$$
G_k \leq \frac{1}{H_k} \frac{2(1+\beta)}{\beta} n\log^2(2\rho mC) + \frac{n\epsilon}{4} + 2n\epsilon.
$$
As $\frac{1}{H_k} = \frac{H_0}{H_k} = \frac{H_0}{H_1}\frac{H_1}{H_2}\dots \frac{H_{k-1}}{H_k} = (1- \frac{h_1}{H_1})^k,$ it follows that $G_k \leq 3\epsilon n$. By construction, $-f_\alpha(\vx_\alpha^{(k+1)}) + f_\alpha(\vx^*_\alpha)\leq 3n\epsilon$, and $\vx_\alpha^{(k+1)}$ is \eqref{eq:alpha-fair-packing}-feasible due to Proposition~\ref{prop:feasibility}.
\end{proof}

\subsection{Convergence Analysis for $\alpha > 1$}\label{sec:alpha>1}
{Recall that in this setting, the algorithm makes updates of the following form:}
\begin{equation*}
   {\vx^{(k)} = \Big(\mI-\frac{\beta(1-\alpha)\diag(\tgrad_r(\vx^{(k-1)}))}{4(1+\alpha\beta)}\Big)\vx^{(k-1)}}
\end{equation*}

Define the vector $\vy^{(k)}$ as:
\begin{equation}\label{eq:a>1-y-k-def}
y_i^{(k)} = (\mA F_{\alpha}(\vx^{(k)}))_i^{{1}/{\beta}} = \big(\mA (\vx^{(k)})^\frac{1}{1-\alpha}\big)_i^{{1}/{\beta}}.
\end{equation}
Clearly, $\vy^{(k)}\geq \zeros.$ Observe that:
\begin{equation}\label{eq:a>1-f-r}
\begin{aligned}
f_r(\vx^{(k)}) {=  - \frac{\innp{\ones, \vx^{(k)}}}{1-\alpha} + \frac{\beta}{1+\beta} \innp{\mA (\vx^{(k)})^\frac{1}{1-\alpha}, \vy^{(k)}}}. 
\end{aligned}
\end{equation}

Recall that the Lagrangian dual of~\eqref{eq:packing-transformed} (and, by the change of variables,~\eqref{eq:alpha-fair-packing}) is $g(\vy) = - \innp{\ones, \vy} + \frac{\alpha}{\alpha - 1}\sum_{j=1}^n (\mA^T \vy)_j^{\frac{\alpha - 1}{\alpha}}$. Interpreting $\vy^{(k)}$ as a dual vector, we can bound the duality gap of a solution ${\vxh^{(k)} = F_\alpha(}\vx^{(k)}{)}$ at any iteration $k$ (using primal feasibility from Proposition~\ref{prop:feasibility}) as:
\begin{equation}\label{eq:a>1-simple-duality-gap}
\begin{aligned}
-f_\alpha(\vxh^{(k)})+f_\alpha(\vx^*_\alpha) = -\frac{\langle{\ones, \vx^{(k)}}\rangle}{1-\alpha} + f_\alpha(\vx^*_\alpha) \leq -\frac{\langle{\ones, \vx^{(k)}}\rangle}{1-\alpha} - g(\vy^{(k)}).
\end{aligned}
\end{equation}
We will assume throughout this section that $\eps \leq \min\big\{\frac{1}{2}, \frac{1}{10(\alpha-1)}\big\}.$
\subsubsection{Regularity Conditions for the Duality Gap}

The next proposition gives a notion of approximate and aggregate complementary slackness, with $\vy^{(k)}$ being interpreted as the vector of dual variables, similar to~\cite{marasevic2015fast}.

\begin{proposition}\label{prop:a>1-appx-comp-slack}
After at most $O({1}/{\beta})$ initial iterations, in every iteration $$\big\langle{\ones, \vy^{(k)}}\big\rangle \leq (1+\epsilon) \big\langle{\mA^T \vy^{(k)}, (\vx^{(k)})^\frac{1}{1-\alpha}}\big\rangle.$$
\end{proposition}
\begin{proof}
First, let us argue that after at most $O(\frac{1}{\beta})$ iterations, there must always exist at least one $i$ with $(\mA (\vx^{(k)})^\frac{1}{1-\alpha})_i \geq 1-\epsilon/2$. Suppose that in any given iteration $\max_{i}(\mA (\vx^{(k)})^{\frac{1}{1-\alpha}})_i \leq 1-\epsilon/4$. Then, as $x_j^{\frac{1}{1-\alpha}}\leq 1$ (by feasibility -- Proposition~\ref{prop:feasibility}) $\forall j,$ $\nabla_j f_r(\vx^{(k)}) \geq \frac{1}{1-\alpha}\left(-1 + Cm\rho (1-\epsilon/4)^{1/\beta}\right) \geq \frac{1}{2(\alpha -1)}$. Hence, each $x_j$ must decrease by a factor at least $1-\frac{\beta(\alpha - 1)}{8(1+\alpha\beta)},$ which means that $(\mA(\vx^{(k)})^\frac{1}{1-\alpha})_i$ increases by a factor at least $(1-\frac{\beta(\alpha -1)}{8(1+\alpha\beta)})^\frac{1}{1-\alpha} \geq 1 + \frac{\beta}{8(1+\alpha\beta)}$. As in any iteration, the most any $(\mA(\vx^{(k)})^\frac{1}{1-\alpha})_i$ can decrease is by a factor at most $1-\beta,$ it follows that after at most initial $O(\frac{1+\alpha\beta}{\beta})$ iterations, it always holds that  $\max_{i}(\mA (\vx^{(k)})^{\frac{1}{1-\alpha}})_i \geq 1-\epsilon/2$.

Let $i^* = \argmax_{i}(\mA (\vx^{(k)})^{\frac{1}{1-\alpha}})_i$ and 
$$
S = \{i: (\mA (\vx^{(k)})^{\frac{1}{1-\alpha}})_i \geq (1-\epsilon/4)(\mA (\vx^{(k)})^{\frac{1}{1-\alpha}})_{i^*}\}.
$$ 
Then, $\forall \ell \notin S,$ $y_\ell^{(k)} \leq (1-\epsilon/4)^{1/\beta}y_{i^*}^{(k)} \leq \frac{\epsilon}{4m}y_{i^*}^{(k)}$. Hence, $\sum_{\ell \notin S}y_{\ell}^{(k)} \leq \frac{\epsilon}{4}y_{i^*}^{(k)} \leq \frac{\epsilon}{4} \sum_{i\in S}y_i^{(k)}$ and we have $\sum_{i\in S}y_i^{(k)} \geq \frac{1}{1+\epsilon/4}\sum_{i'=1}^m y_{i'}^{(k)}$. It follows that:
\begin{align*}
\big\langle{\vy^{(k)},\mA(\vx^{(k)})^{\frac{1}{1-\alpha}}}\big\rangle
&\geq \sum_{i \in S}y_i^{(k)}(\mA (\vx^{(k)})^{\frac{1}{1-\alpha}})_i 
\geq (1-\epsilon/2)(1-\epsilon/4)\sum_{i \in S}y_i^{(k)}\\
&\geq \frac{(1-\epsilon/2)(1-\epsilon/4)}{1+\epsilon/4}\big\langle{\ones, \vy^{(k)}}\big\rangle.
\end{align*}
The rest of the proof is by $\frac{1+\epsilon/4}{(1-\epsilon/2)(1-\epsilon/4)}\leq 1+\epsilon$.
\end{proof}

To construct and use the same argument as before (namely, to guarantee that $H_k G_k \leq H_{k-1}G_{k-1} + O(\epsilon)(1-\alpha)f_{\alpha}(\vx^*_\alpha)$ for some  gap $G_k$), we need to ensure that the argument can be started from a gap $G_0 = O(1) (1-\alpha)f_\alpha(\vx^*_\alpha)$. The following lemma gives sufficient conditions for ensuring constant multiplicative gap. When those conditions are not met, we show that $f_r(\vx^{(k)})$ must decrease multiplicatively (Lemma~\ref{lemma:a>1-mul-dec}), which guarantees that there cannot be many such iterations. 
Define:
\begin{gather*}
S_{+} \defeq \big\{j: (x_j^{(k)})^{\frac{\alpha}{1-\alpha}}(\mA^T \vy^{(k)})_j \geq 1 + \frac{1}{10(\alpha -1)}\big\},\\
S_{-} \defeq \big\{j: (x_j^{(k)})^{\frac{\alpha}{1-\alpha}}(\mA^T \vy^{(k)})_j \leq 1 - \frac{1}{10}\big\}.
\end{gather*}
The next lemma gives sufficient conditions for $\vx^{(k)}$ to have a constant relative error. 
\begin{lemma}\label{lemma:a>1-const-init-gap}
After the initial $O(\frac{1}{\beta})$ iterations, if all following conditions hold:
\begin{enumerate}
\item $-\sum_{j \in S_{+}}x_j^{(k)} \big(1- (x_j^{(k)})^{\frac{\alpha}{1-\alpha}}(\mA^T \vy^{(k)})_j\big) \leq \frac{1}{10(\alpha -1)}\innp{\ones, \vx^{(k)}};$ 
\item $\sum_{j \in S_{-}}x_j^{(k)} \big(1- (x_j^{(k)})^{\frac{\alpha}{1-\alpha}}(\mA^T \vy^{(k)})_j\big) \leq \frac{1}{10}\innp{\ones, \vx^{(k)}};$ and
\item $\big\langle{\vy^{(k)}, \mA (\vx^{(k)})^\frac{1}{1-\alpha}}\big\rangle \leq 2\innp{\ones, \vx^{(k)}}$
\end{enumerate}
then $f_r(\vx^{(k)}) + f_\alpha(\vx^*_\alpha) \leq -2 f_\alpha(\vx_\alpha^*)$.
\end{lemma}
\begin{proof}
Denote $\Delta_j = (x_j^{(k)})^{\frac{\alpha}{1-\alpha}}(\mA^T\vy^{(k)})_j.$ Let us start by bounding the true duality gap (using feasibility from Proposition~\ref{prop:feasibility} and approximate complementary slackness from Proposition~\ref{prop:a>1-appx-comp-slack}):
\begin{align}
\frac{\innp{\ones, \vx^{(k)}}}{\alpha - 1}& + f_\alpha(\vx^*_\alpha)%
\leq \frac{\innp{\ones, \vx^{(k)}}}{\alpha - 1} - g(\vy^{(k)})\notag\\
&\leq \frac{\innp{\ones, \vx^{(k)}}}{\alpha - 1} + (1+\epsilon)\innp{\mA^T\vy^{(k)}, (\vx^{(k)})^{\frac{1}{1-\alpha}}} - \frac{\alpha}{\alpha -1}\sum_{j=1}^n (\mA^T\vy^{(k)})_j^{\frac{\alpha - 1}{\alpha}}\notag\\
&= \frac{1}{\alpha -1}\sum_{j=1}^n x_j^{(k)}\left(1 + (\alpha - 1)\Delta_j -\alpha \Delta_j^{\frac{\alpha -1}{\alpha}}\right) + \epsilon \innp{\mA^T\vy^{(k)}, (\vx^{(k)})^{\frac{1}{1-\alpha}}}\notag \\
&= \sum_{j=1}^n \xi_j + \epsilon \innp{\mA^T\vy^{(k)}, (\vx^{(k)})^{\frac{1}{1-\alpha}}},\label{eq:a>1-simple-init-gap}
\end{align}
where $\xi_j = \frac{x_j^{(k)}\left(1 + (\alpha - 1)\Delta_j -\alpha \Delta_j^{{(\alpha -1)}/{\alpha}}\right)}{\alpha -1}$. 
To bound the expression from~\eqref{eq:a>1-simple-init-gap}, we will split the sum $\sum_{j=1}^n \xi_j$ into two: corresponding to terms with $\Delta_j \geq 1$ and corresponding to terms with $\Delta_j < 1$. 
For the former, as $\Delta_j^{\frac{\alpha-1}{\alpha}}\geq 1,$ we have:
\begin{align}
\sum_{j: \Delta_j \geq 1} \xi_j &\leq \frac{1}{\alpha -1}\sum_{j: \Delta_j \geq 1} x_j^{(k)}\left(1 + (\alpha - 1)\Delta_j -\alpha\right)\notag \\
&= \sum_{j: 1\leq \Delta_j \leq 1 + \frac{1}{10(\alpha-1)}}x_j^{(k)}(\Delta_j - 1) + \sum_{j \in S_+}x_j^{(k)}(\Delta_j - 1)\notag\\
&\leq \frac{1}{5(\alpha-1)}\innp{\ones, \vx^{(k)}},\label{eq:a>1-init-gap-pos}
\end{align}
where the last inequality is by $\sum_{j: 1\leq \Delta_j \leq 1 + \frac{1}{10(\alpha-1)}}x_j^{(k)}(\Delta_j - 1) \leq \frac{\innp{\ones, \vx^{(k)}}}{10(\alpha-1)}$ and the first condition from the statement of the lemma. 

Consider now the terms with $\Delta_j < 1.$ %
As $\Delta_j^{\frac{\alpha-1}{\alpha}} \geq \Delta_j:$
\begin{align}
\sum_{j: \Delta_j < 1}\xi_j &\leq \frac{1}{\alpha -1 }\sum_{j: \Delta_j < 1} x_j^{(k)}\big(1 + (\alpha-1)\Delta_j - \alpha \Delta_j\big)\notag\\
&=\frac{1}{\alpha-1}\sum_{j: 1- \frac{1}{10} < \Delta_j < 1} x_j^{(k)}(1-\Delta_j) + \frac{1}{\alpha-1}\sum_{j\in S_{-}} x_j^{(k)}(1-\Delta_j)\notag\\
&\leq \frac{1}{5(\alpha-1)}\innp{\ones, \vx^{(k)}}.\label{eq:a>1-init-gap-neg}
\end{align}

The third condition from the statement of the lemma guarantees that 
$$
\epsilon \innp{\mA^T\vy^{(k)}, (\vx^{(k)})^{\frac{1}{1-\alpha}}} \leq 2\epsilon\innp{\ones, \vx^{(k)}}\leq \frac{1}{5(\alpha-1)}\innp{\ones, \vx^{(k)}},
$$
as $\eps \leq \frac{1}{10(\alpha-1)}$. Hence, combining~\eqref{eq:a>1-simple-init-gap}-\eqref{eq:a>1-init-gap-neg}:
$%
\frac{\innp{\ones, \vx^{(k)}}}{\alpha - 1} + f_\alpha(\vx^*_\alpha) \leq \frac{3}{5(\alpha-1)}\innp{\ones, \vx^{(k)}}.%
$ 
Equivalently: $\frac{\innp{\ones, \vx^{(k)}}}{\alpha - 1} \leq -\frac{5}{2}f_\alpha(\vx_\alpha^*).$ {Using Eq.~\eqref{eq:a>1-f-r}, $\mA (\vx^{(k)})^{\frac{1}{1-\alpha}} \leq 1$ (by feasibility -- Proposition~\ref{prop:feasibility}), and } the third condition in the lemma statement, 
\begin{align*}
f_r(\vx^{(k)}) &{= \frac{\innp{\ones, \vx^{(k)}}}{\alpha-1} + \frac{\beta}{1-\beta}\innp{\mA (\vx^{(k)})^{\frac{1}{1-\alpha}}, \vy^{(k)}}}\\
&\leq \frac{\innp{\ones, \vx^{(k)}}}{\alpha-1}\Big(1 + \frac{2\beta(\alpha-1)}{1+\beta}\Big)\leq \frac{\innp{\ones, \vx^{(k)}}}{\alpha-1}\Big(1 + \frac{\eps(\alpha-1)}{2}\Big) \leq \frac{21}{20}\frac{\innp{\ones, \vx^{(k)}}}{\alpha-1}, 
\end{align*}
as $\beta \leq \eps/4$ and $\eps \leq \frac{1}{10(\alpha-1)}.$ Putting everything together:
\begin{align*}
f_r(\vx^{(k)}) + f_\alpha(\vx^*_\alpha) \leq \frac{13}{20(\alpha-1)}\innp{\ones, \vx^{(k)}}\leq - \frac{5}{2}\cdot \frac{13}{20}f_\alpha(\vx^*_\alpha) \leq -2f_\alpha(\vx^*_\alpha). 
\end{align*}
\end{proof}
\begin{lemma}\label{lemma:a>1-mul-dec}
If in iteration $k$ any of the conditions from Lemma~\ref{lemma:a>1-const-init-gap} does not hold, then $f_r(\vx^{(k)})$ must decrease by a factor at most 
$$
\max\Big\{1-{\theta(\beta(\alpha -1))},\; 1-\theta(\beta) \min\Big\{\frac{1}{10(\alpha -1)},\, 1\Big\}\Big\}.
$$
\end{lemma}
\begin{proof}
If the conditions from Lemma~\ref{lemma:a>1-const-init-gap} do not hold, then we must have (at least) one of the following cases.

\noindent\textbf{Case 1:} $-\sum_{j \in S_{+}}x_j^{(k)} \big(1- (x_j^{(k)})^{\frac{\alpha}{1-\alpha}}(\mA^T \vy^{(k)})_j\big) > \frac{\innp{\ones, \vx^{(k)}}}{10(\alpha-1)}.$ Observe that, by the definition of $S_{+},$ for all $j \in S_{+},$ $\tgradj_r(\vx^{(k)}) \geq \min\{\frac{1}{10(\alpha-1)}, 1\}$ and $(1-\alpha)\nabla_j f_r(\vx^{(k)})\geq \frac{1}{10(\alpha-1)} > 0$.  From Lemma~\ref{lemma:smoothness}{, as $x_j^{(k)}\nabla_j f(\vx^{(k)})\tgradj_r(\vx^{(k)}) \geq 0,$ $\forall j$}:
\begin{align*}
f(\vx^{(k+1)})& - f(\vx^{(k)})\\
&\leq - \frac{\beta(1 -\alpha)}{8(1 + \alpha \beta)}\sum_{j \in S_{+}}x_j^{(k)}\nabla_j f(\vx^{(k)})\tgradj_r(\vx^{(k)})\\
&\leq  \min\Big\{\frac{1}{10(\alpha-1)}, 1\Big\}\frac{\beta}{8(1 + \alpha \beta)}\sum_{j \in S_{+}}x_j^{(k)}\left(1- (x_j^{(k)})^{\frac{\alpha}{1-\alpha}}(\mA^T \vy^{(k)})_j\right)\\
&\leq -\min\Big\{\frac{1}{10(\alpha-1)}, 1\Big\}\frac{\beta}{80(\alpha-1) (1+\alpha\beta)}\innp{\ones, \vx^{(k)}}.
\end{align*}
Assume that $\big\langle{\vy^{(k)}, \mA (\vx^{(k)})^{\frac{1}{1-\alpha}}}\big\rangle\leq 2\innp{\ones, \vx^{(k)}}$ (otherwise we would have Case 3 below). Then $f_r(\vx^{(k)})\leq \big(\frac{1}{\alpha-1} + \frac{2\beta}{1+\beta}\big)\innp{\ones, \vx^{(k)}},$ and, hence 
\begin{align*}
\innp{\ones, \vx^{(k)}}
&\geq \Big(\frac{1}{\alpha-1} + \frac{2\beta}{1+\beta}\Big)^{-1}f_r(\vx^{(k)}) %
\geq \frac{\alpha-1}{2}f_r(\vx^{(k)}).
\end{align*}
Therefore, it follows that $f(\vx^{(k+1)}) - f(\vx^{(k)}) \leq -\theta\big({\beta}\min\big\{\frac{1}{10(\alpha-1)}, 1\big\}\big)f_r(\vx^{(k)}).$

\noindent\textbf{Case 2:} $\sum_{j \in S_{-}}x_j^{(k)} \big(1- (x_j^{(k)})^{\frac{\alpha}{1-\alpha}}(\mA^T \vy^{(k)})_j\big) > \frac{1}{10}\innp{\ones, \vx^{(k)}}.$ Observe that, by the definition of $S_{-},$ for all $j \in S_{-},$ $\tgradj_r(\vx^{(k)}) \leq - \frac{1}{10}$ and $(1-\alpha)\nabla_j f_r(\vx^{(k)})\leq - \frac{1}{10} <0.$ From Lemma~\ref{lemma:smoothness}:
\begin{align*}
f(\vx^{(k+1)}) - f(\vx^{(k)}) &\leq - \frac{\beta(1 -\alpha)}{8(1 + \alpha \beta)}\sum_{j \in S_{-}}x_j^{(k)}\nabla_j f(\vx^{(k)})\tgradj_r(\vx^{(k)})\\
&\leq - \frac{\beta}{80(1 + \alpha \beta)}\sum_{j \in S_{-}}x_j^{(k)} \left(1- (x_j^{(k)})^{\frac{\alpha}{1-\alpha}}(\mA^T \vy^{(k)})_j\right)\\
&< - \frac{\beta}{800(1 + \alpha \beta)}\innp{\ones, \vx^{(k)}}. 
\end{align*}
Similar as in the previous case, assume that $\innp{\vy^{(k)}, \mA (\vx^{(k)})^{\frac{1}{1-\alpha}}}\leq 2\innp{\ones, \vx^{(k)}}$. Then $\innp{\ones, \vx^{(k)}}\geq \frac{\alpha-1}{2}f_r(\vx^{(k)})$, and we have $f_r(\vx^{(k+1)})-f_r(\vx^{(k)})\leq - \theta(\beta(\alpha-1))f_r(\vx^{(k)}).$ 

\noindent\textbf{Case 3:} $\big\langle{\vy^{(k)}, \mA (\vx^{(k)})^\frac{1}{1-\alpha}}\big\rangle \geq 2\innp{\ones, \vx^{(k)}}$. Equivalently: $\frac{1}{2}\big\langle{\vy^{(k)}, \mA (\vx^{(k)})^\frac{1}{1-\alpha}}\big\rangle \geq \innp{\ones, \vx^{(k)}}$. Subtracting $\big\langle{\vy^{(k)}, \mA (\vx^{(k)})^\frac{1}{1-\alpha}}\big\rangle$ from both sides and rearranging the terms: 
\begin{equation}\label{eq:gap-it-1}
\sum_{j=1}^n x_j^{(k)}\big(-1 + (x_j^{(k)})^\frac{\alpha}{1-\alpha}(\mA^T\vy^{(k)})_j\big) \geq \frac{1}{2} \innp{\vy^{(k)}, \mA (\vx^{(k)})^\frac{1}{1-\alpha}}. 
\end{equation}
Let $\zeta_j = -1 + (x_j^{(k)})^\frac{\alpha}{1-\alpha}(\mA^T\vy^{(k)})_j$. Then:
\begin{align}
\sum_{j=1}^n x_j^{(k)}\big(-1 + (x_j^{(k)})^\frac{\alpha}{1-\alpha}(\mA^T\vy^{(k)})_j\big) 
&\leq \frac{1}{2}\innp{\ones, \vx^{(k)}} + \sum_{j:\zeta_j > 1/2}x_j^{(k)}\zeta_j\notag\\
&\leq \frac{1}{4}\innp{\vy^{(k)}, \mA (\vx^{(k)})^\frac{1}{1-\alpha}} + \sum_{j:\zeta_j > 1/2}x_j^{(k)}\zeta_j. \label{eq:gap-it-2}
\end{align}
As $f_r(\vx^{(k)})\leq \big(\frac{1}{2(\alpha -1)} + \frac{\beta}{1+\beta}\big)\big\langle{\vy^{(k)}, \mA (\vx^{(k)})^\frac{1}{1-\alpha}}\big\rangle$, combining~\eqref{eq:gap-it-1} and~\eqref{eq:gap-it-2}:
\begin{equation}\label{eq:gap-it-3}
\sum_{j:\zeta_j > 1/2}x_j^{(k)}\zeta_j \geq \frac{1}{4}\innp{\vy^{(k)}, \mA (\vx^{(k)})^\frac{1}{1-\alpha}} \geq \frac{1}{4}\left(\frac{1}{2(\alpha -1)} + \frac{\beta}{1+\beta}\right)^{-1}f_r(\vx^{(k)}).
\end{equation}
Using Lemma~\ref{lemma:smoothness}, it follows that, 
$
f_r(\vx^{(k+1)})-f_r(\vx^{(k)})\leq - \frac{\beta}{16(1+\alpha\beta)}\sum_{j:\zeta_j > 1/2}x_j^{(k)}\zeta_j,
$  
which, combined with~\eqref{eq:gap-it-3}, gives: 
$
f_r(\vx^{(k+1)})\leq \big(1-\theta({\beta(\alpha-1)})\big)f_r(\vx^{(k)})
$
.
\end{proof}

\subsubsection{The Decrease in the Duality Gap and the Convergence Bound}\label{sec:a>1-gap-decrease}

Using Lemma~\ref{lemma:a>1-mul-dec}, within the first $O(\frac{1}{\beta}+\frac{1}{\beta}\max\{\frac{1}{\alpha-1}, \alpha-1\}\log(\frac{f_r(\vx^{(0)})}{f_r(\vx^*_r)}))$ iterations, there must exist at least one iteration in which the conditions from Proposition~\ref{prop:a>1-appx-comp-slack} and Lemma~\ref{lemma:a>1-const-init-gap} hold. With the (slight) abuse of notation, we treat first such iteration as our initial ($k=0$) iteration, and focus on proving the convergence over a subsequence of iterations that come after it. We call the iterations over which we perform the gap analysis the ``gap iterations'' and we define them as iterations in which:
\begin{equation}\label{eq:gap-iterations}
\big\langle{\vy^{(k)}, \mA (\vx^{(k)})^\frac{1}{1-\alpha}}\big\rangle \leq {2}\big\langle{\ones, \vx^{(k)}}\big\rangle.
\end{equation}
Due to Lemma~\ref{lemma:a>1-mul-dec}, in non-gap iterations, $f_r(\vx^{(k)})$ must decrease multiplicatively. 
Hence, we focus only on the gap iterations, which we index by $k$ below.

To construct $G_k$, we define the upper bound to be $U_k = f_r(\vx^{(k+1)})$. The lower bound is simply defined through the use of the Lagrangian dual as  $L_k = \frac{\sum_{\ell=0}^k h_\ell g(\vy^{(\ell)})}{H_k}$. 
\paragraph{Initial gap} Using Lemma~\ref{lemma:smoothness}, $U_0 = f_r(\vx^{(1)})\leq f_r(\vx^{(0)})$. Thus, by Lemma~\ref{lemma:a>1-const-init-gap} and the choice of the initial point $k=0$ described above, we have:
\begin{equation}\label{eq:a>1-init-gap}
G_0 = U_0 - L_0 \leq -2f_\alpha(\vx^*_\alpha).
\end{equation}
\paragraph{The gap decrease} The next step is to show that, for a suitably chosen sequence $\{h_k\}_k,$ $H_k G_k - H_{k-1}G_{k-1}\leq O(\epsilon)(1-\alpha)f_\alpha(\vx^*_\alpha)$. This would immediately imply $G_k \leq \frac{H_0 G_0}{H_k} + O(\epsilon)(1-\alpha)f_\alpha(\vx^*_\alpha)$ which is $= O(\epsilon)(1-\alpha)f_\alpha(\vx^*_\alpha)$ when $H_0/H_k = O(\epsilon(\alpha -1))$, due to the bound on the initial gap~\eqref{eq:a>1-init-gap}. As $U_k = f_r(\vx^{(k+1)}) \geq \frac{\innp{\ones, \vx^{(k+1)}}}{\alpha-1}$ and $L_k \geq -f_\alpha(\vx^*_\alpha)$, taking $\vxh^{(k)} = (\vx^{(k+1)})^\frac{1}{1-\alpha}$, it would immediately follow that:
\begin{align*}
-f_{\alpha}(\vxh^{(k)}) + f_\alpha(\vx^*_\alpha) \leq  O(\epsilon)(1-\alpha)f_\alpha(\vx^*_\alpha).
\end{align*}
Since $\vxh^{(k)}$ is \eqref{eq:alpha-fair-packing}-feasible (Prop.~\ref{prop:feasibility}),  $\vxh^{(k)}$ is an $O(\epsilon)$-approximate solution to \eqref{eq:alpha-fair-packing}.

To bound $H_kG_k - H_{k-1}G_{k-1},$ we will need the following technical proposition that bounds $H_kL_k - H_{k-1}L_{k-1}$ (the change in the lower bound).

\begin{proposition}\label{prop:a>1-change-in-Lk}
For any two consecutive gap iterations $k-1, k$:
\begin{align*}
H_kL_k & - H_{k-1}L_{k-1}\\
\geq & h_k\big[f_r(\vx^{(k)})-\big\langle{\nabla f_r(\vx^{(k)}), \vx^{(k)}}\big\rangle - 8\epsilon(\alpha-1) f_\alpha(\vx^*_\alpha)\\
&+ \frac{\alpha}{\alpha-1}\sum_{j=1}^n x_j^{(k)}\big(\big(1 +\tgradj_r(\vx^{(k)})\big)^{\frac{\alpha-1}{\alpha}}- \big(1 +(1-\alpha)\nabla_jf_r(\vx^{(k)})\big)\big)\big].
\end{align*}
\end{proposition}
\begin{proof}
By the definition of the lower bound:
\begin{equation}\label{eq:a>1-lb-change-1}
H_k L_k - H_{k-1}L_{k-1} = h_kg(\vy^{(k)}) = h_k \big(-\innp{\ones, \vy^{(k)}} + \frac{\alpha}{\alpha -1}\sum_{j=1}^n (\mA^T\vy^{(k)})_j^{\frac{\alpha-1}{\alpha}}\big).
\end{equation}
From Proposition~\ref{prop:a>1-appx-comp-slack}:
$\big\langle{\ones, \vy^{(k)}}\big\rangle \leq (1+\epsilon)\big\langle{\mA^T\vy^{(k)}, (\vx^{(k)})^{\frac{1}{1-\alpha}}}\big\rangle,$ while from   {Eq.~\eqref{eq:a>1-f-r}}:
$$
    {f_r(\vx^{(k)})-\innp{\nabla f_r(\vx^{(k)}), \vx^{(k)}} = \big(\frac{\beta}{1+\beta} + \frac{1}{\alpha -1}\big)\big\langle\mA^T\vy^{(k)}, (\vx^{(k)})^{\frac{1}{1-\alpha}}\big\rangle.}
$$
Hence, %
\begin{align*}
\big\langle{\ones, \vy^{(k)}}\big\rangle \leq & (1+\epsilon)\big\langle{\mA^T\vy^{(k)}, (\vx^{(k)})^{\frac{1}{1-\alpha}}}\big\rangle\\
=& -f_r(\vx^{(k)}) + \big\langle{\nabla f_r(\vx^{(k)}), \vx^{(k)}}\big\rangle + \frac{\alpha}{\alpha -1}\big\langle{\mA^T\vy^{(k)}, (\vx^{(k)})^{\frac{1}{1-\alpha}}}\big\rangle\\
& + \Big(\epsilon +\frac{\beta}{1+\beta} \Big) \big\langle{\mA^T\vy^{(k)}, (\vx^{(k)})^{\frac{1}{1-\alpha}}}\big\rangle.
\end{align*}
Since $k$ is a gap iteration, $f_r(\vx^{(k)})\geq \big(\frac{1}{2(\alpha-1)} + \frac{\beta}{1+\beta}\big)\big\langle{\mA^T\vy^{(k)}, (\vx^{(k)})^{\frac{1}{1-\alpha}}}\big\rangle$. Hence, %
\begin{align}
\big\langle{\ones, \vy^{(k)}}\big\rangle \leq &  -f_r(\vx^{(k)}) + \big\langle{\nabla f_r(\vx^{(k)}), \vx^{(k)}}\big\rangle\notag \\
&+ \frac{\alpha}{\alpha -1}\big\langle{\mA^T\vy^{(k)}, (\vx^{(k)})^{\frac{1}{1-\alpha}}}\big\rangle
 + \frac{10}{4}(\alpha-1)\epsilon  f_r(\vx^{(k)})\notag\\
 \leq & -f_r(\vx^{(k)}) + \big\langle{\nabla f_r(\vx^{(k)}), \vx^{(k)}}\big\rangle\notag\\
 &+ \frac{\alpha}{\alpha -1}\big\langle{\mA^T\vy^{(k)}, (\vx^{(k)})^{\frac{1}{1-\alpha}}}\big\rangle
 - 8(\alpha-1)\epsilon f_\alpha(\vx^*_\alpha),\label{eq:a>1-lb-change-2}
\end{align}
where the last inequality follows from $f_r(\vx^{(k)})\leq f_r(\vx^{(0)})$ (as $f_r(\cdot)$ decreases in each iteration) and $f(\vx^{(0)})\leq - \frac{11}{4}f_\alpha(\vx^*_\alpha)$ (by the choice of $\vx^{(0)}$ and Lemma~\ref{lemma:a>1-const-init-gap}). Combining~\eqref{eq:a>1-lb-change-1} and~\eqref{eq:a>1-lb-change-2}:
\begin{align*}
H_kL_k - H_{k-1}L_{k-1} \geq & h_k \Big(f_r(\vx^{(k)})-\big\langle{\nabla f_r(\vx^{(k)}), \vx^{(k)}}\big\rangle +  8\epsilon(\alpha-1)f_\alpha(\vx^*_\alpha)\\
&
+\frac{\alpha}{\alpha-1}\sum_{j=1}^n \big((\mA^T\vy^{(k)})_j^{\frac{\alpha-1}{\alpha}}-(x_j^{(k)})^\frac{1}{1-\alpha}(\mA^T\vy^{(k)})_j\big)\Big).
\end{align*}
Finally, as $(\mA^T\vy^{(k)})_j = (x_j^{(k)})^\frac{-\alpha}{1-\alpha}(1+(1-\alpha)\nabla_j f_r(\vx^{(k)}))$ and $(1-\alpha)\nabla_j f_r(\vx^{(k)})\geq \tgradj_r(\vx^{(k)})$, the statement of the proposition follows. %
\end{proof}
\begin{lemma}\label{lemma:a>1-gap-decrease}
If, for $k \geq 1,$ $\frac{h_k}{H_k} \leq \frac{\beta\min\{\alpha-1, 1\}}{16(1+\alpha\beta)},$ then 
$$
H_kG_k - H_{k-1}G_{k-1} \leq -8 h_k \epsilon (\alpha-1)f_\alpha(\vx^*_\alpha).
$$
\end{lemma}
\begin{proof}
Using Lemma~\ref{lemma:smoothness} (and as $f_r(\vx^{(k)})$ decreases by the Lemma~\ref{lemma:smoothness} guarantees regardless of whether the iteration is a gap iteration or not):
\begin{equation}
H_k U_k - H_{k-1}U_{k-1} \leq h_k f_r(\vx^{(k)}) - H_k \frac{\beta(1-\alpha)}{8(1+\alpha\beta)}\sum_{j=1}^n x_j^{(k)}\nabla_j f_r(\vx^{(k)})\tgradj_r(\vx^{(k)}).\notag
\end{equation}
Combining with the change in the lower bound from Proposition~\ref{prop:a>1-change-in-Lk}, it follows that to prove the statement of the lemma it suffices to show that, $\forall j$:
\begin{align*}
\xi_j \defeq & h_k \big[\nabla_j f_r(\vx^{(k)}) - \frac{\alpha}{\alpha-1}\big(\big(1 +\tgradj_r(\vx^{(k)})\big)^{\frac{\alpha-1}{\alpha}}- \big(1 +(1-\alpha)\nabla_jf_r(\vx^{(k)})\big)\big)\big]\\
&- H_k \frac{\beta(1-\alpha)}{8(1+\alpha\beta)}\nabla_j f_r(\vx^{(k)})\tgradj_r(\vx^{(k)})
\leq  0.
\end{align*}
Consider the following three cases:

\noindent\textbf{Case 1:} $(1-\alpha)\nabla_j f_r(\vx^{(k)})\in [-1/2, 1]$. Then $\tgradj_r(\vx^{(k)}) = (1-\alpha)\nabla_j f_r(\vx^{(k)}).$ A simple corollary of Taylor's Theorem is that in this setting:
\begin{equation}\label{eq:a>1-taylor-ineq}
\big(1 +\tgradj_r(\vx^{(k)})\big)^{\frac{\alpha-1}{\alpha}} \geq 1 + \frac{\alpha-1}{\alpha}\tgradj_r(\vx^{(k)}) - \frac{\alpha-1}{\alpha^2} (\tgradj_r(\vx^{(k)}))^2.
\end{equation}

Using Eq.~\eqref{eq:a>1-taylor-ineq} from above:
\begin{align*}
\xi_j \leq & h_k \Big[ \frac{\tgradj_r(\vx^{(k)})}{1-\alpha} - \frac{\alpha}{\alpha-1} \Big(-\frac{1}{\alpha}\tgradj_r(\vx^{(k)}) - \frac{\alpha-1}{\alpha^2} \big(\tgradj_r(\vx^{(k)})\big)^2\Big)\Big]\\
&- \frac{H_k \beta}{8(1+\alpha\beta)}(\tgradj_r(\vx^{(k)}))^2\\
=& (\tgradj_r(\vx^{(k)}))^2 \Big(\frac{h_k}{\alpha} - \frac{H_k \beta}{8(1+\alpha\beta)}\Big).
\end{align*}
As $\frac{h_k}{H_k}\leq \frac{\beta}{8(1+\alpha\beta)}\leq \frac{\beta\alpha}{8(1+\alpha\beta)}$, it follows that $\xi_j \leq 0$.
\\[10pt]
\noindent\textbf{Case 2:} $(1-\alpha)\nabla_j f_r(\vx^{(k)})\in [-1, -1/2).$ Then $\tgradj_r(\vx^{(k)}) = (1-\alpha)\nabla_j f_r(\vx^{(k)})$ and $|\tgradj_r(\vx^{(k)})|>\frac{1}{2}.$ As in this case $\big(1 +\tgradj_r(\vx^{(k)})\big)^{\frac{\alpha-1}{\alpha}} \geq 1 +\tgradj_r(\vx^{(k)}),$ we have:
\begin{align*}
\xi_j \leq \nabla_j f_r(\vx^{(k)})\Big(h_k - H_k \frac{\beta(\alpha-1)}{16(1+\alpha\beta)}\Big),
\end{align*}
which is $\leq 0,$ as $\frac{h_k}{H_k}\leq \frac{\beta\min\{\alpha-1, 1\}}{16(1+\alpha\beta)}$ and $\nabla_j f_r(\vx^{(k)}) > 0$.%
\\[10pt]
\noindent\textbf{Case 3:} $(1-\alpha)\nabla_j f_r(\vx^{(k)})>1$. Then $\tgradj_r(\vx^{(k)})=1$, and we have:
\begin{align*}
\xi_j \leq & h_k\big[\nabla_jf_r(\vx^{(k)})- \frac{\alpha}{\alpha-1}\big(2^{\frac{\alpha-1}{a}} -1 -(1-\alpha)\nabla_j f_r(\vx^{(k)})\big)\big]\\
&- H_k \frac{\beta(1-\alpha)}{8(1+\alpha\beta)}\nabla_j f_r(\vx^{(k)})\\
\leq & (1-\alpha)\nabla_j f_r(\vx^{(k)}) \Big(h_k - \frac{H_k\beta}{8(1+\alpha\beta)}\Big),
\end{align*}
which is non-positive, as $\frac{h_k}{H_k}\leq \frac{\beta}{8(1+\alpha\beta)}.$
\end{proof}

We can now state the final convergence bound.

\begin{theorem}\label{thm:a>1-convergence}
Given $\epsilon  \in (0, \min\{1/2, {1}/{(10(\alpha-1))}\}]$, after at most 
$$O\Big(\max\Big\{\frac{\alpha^3\log(n\rho)\log(mn\rho/\epsilon)}{\epsilon}, \frac{\log(\frac{1}{\epsilon(\alpha-1)})\log(mn\rho/\epsilon)}{\epsilon(\alpha-1)}\Big\}\Big)$$ 
iterations of \textsc{FairPacking}, 
$$
f_\alpha(\vx_\alpha^{(k+1)})-f_\alpha(\vx^*_\alpha)\geq 10\epsilon(\alpha-1)f_\alpha(\vx^*_\alpha),
$$ 
where $\vx_\alpha^{(k+1)} = (\vx^{(k+1)})^\frac{1}{1-\alpha}$.
\end{theorem}
\begin{proof}
At initialization, $f_r(\cdot)$ takes value less than $\frac{n(3n\rho)^{\alpha-1}}{\alpha-1}$ and decreases in every subsequent iteration. From Proposition~\ref{prop:packing-bounds-on-opt}, $-f_\alpha(\vx^*_\alpha)\geq \frac{n}{\alpha-1}$. As $f_r(\vx) \geq \frac{\innp{\ones, \vx}}{\alpha-1}$ and the algorithm always maintains solutions $\vx^{(k)}$ that are feasible in \eqref{eq:packing-transformed}, $\min_k f_r(\vx^{(k)}) \geq -f_\alpha(\vx^*_\alpha)\geq \frac{n}{\alpha-1}$. 
Using Proposition~\ref{prop:a>1-appx-comp-slack} and Lemma~\ref{lemma:a>1-mul-dec}, there are at most $O\big(\frac{1}{\beta}\max\big\{\frac{1}{\alpha-1}, \alpha-1\big\}{(\alpha-1)\log(n\rho)}\big) = O\big(\frac{(1+\alpha)\max\{(\alpha-1)^2, 1\}\log(n\rho)\log(mn\rho/\epsilon)}{\epsilon}\big)$ non-gap iterations before $f_r(\cdot)$ reaches its minimum value. Using the second part of Proposition~\ref{prop:regularization}, if this happens, it follows that $f_\alpha(\vxh^{(k+1)})-f_\alpha(\vx^*_\alpha)\geq - 2\epsilon(1-\alpha)f_\alpha(\vx^*)$, and we are done. 
For the gap iterations, choose $h_0 = H_0 = 1$, $\frac{h_\ell}{H_\ell} = (1- \frac{H_{\ell-1}}{H_\ell}) =\frac{\beta\min\{\alpha-1, 1\}}{16(1+\alpha\beta)},$ for $\ell \geq 1.$ Using Lemma~\ref{lemma:a>1-gap-decrease}: \begin{align*}
G_k &\leq \frac{H_0G_0}{H_k} - 8\epsilon(\alpha-1)f_\alpha(\vx^*_\alpha)%
= \Big(1 - \frac{\beta\min\{\alpha-1, 1\}}{16(1+\alpha\beta)}\Big)^k G_0  - 8\epsilon(\alpha-1)f_\alpha(\vx^*_\alpha).
\end{align*}
As $G_0 \leq -2f_\alpha(\vx^*_\alpha)$, after $k \geq \frac{\log(\frac{1}{\epsilon(\alpha-1)})}{\beta\min\{\alpha-1, 1\}}16(1+\alpha\beta) = O\big(\frac{(1+\alpha)\log(\frac{1}{\epsilon(\alpha-1)})\log(mn\rho/\epsilon)}{\epsilon\min\{\alpha-1, 1\}}\big)$ iterations, it must be $-f_\alpha(\vxh^{(k+1)}) + f_\alpha(\vx^*_\alpha) %
\leq G_k \leq 10 \epsilon(1-\alpha)f_\alpha(\vx^*_\alpha)$, as claimed.
\end{proof}
%
%
%
\section{Fair Covering}\label{se:covering}
In this section, we show how to reduce the fair covering problem to the $\alpha < 1$ case from Section~\ref{sec:alpha<1}. We will be assuming throughout that $\beta \geq \frac{\epsilon/4}{\log(mn\rho/\epsilon)},$ as otherwise the problem can be reduced to the linear covering (see, e.g.,~\cite{LP-jelena-lorenzo}). Note that the only aspect of the analysis that relies on $\beta$ being ``sufficiently small'' in the $\alpha\in [0, 1)$ case is to ensure that $f_r$ closely approximates $-f_\alpha$ around the optimum of~\eqref{eq:alpha-fair-packing}, \eqref{eq:packing-transformed}. Here, we will need to choose $\beta'$ to be ``sufficiently small'' to ensure that the lower bound from the $\alpha<1$ case closely approximates $-g_\beta$ around the optimum $\vy^*.$ Since we do not need to ensure the feasibility of the packing problem, in this section we take $C=1,$ so that $f_r(\vx) = - \innp{\ones, \vx} + \frac{\beta}{1+\beta}\sum_{i=1}^m (\mA\vx)_i^{{(1+\beta)}/{\beta}}.$ As before, the upper bound is defined as $U_k = f_r(\vx^{(k+1)}).$ 
The lower bound $L_k$ is the same as the one from Section~\ref{sec:alpha<1}, with the choice of $\beta'$ as in Algorithm~\ref{algo:covering} (\textsc{FairCovering}). 
\begin{algorithm}
\caption{\textsc{FairCovering}($\mA, \eps, \beta$)}\label{algo:covering}
\begin{algorithmic}[1]
\State If $\beta \leq 0$, set $\beta = \frac{\epsilon/4}{\log({mn\rho}/{\epsilon})}$. 
 Initialize: $\vx^{(0)}_j = \frac{1}{n\rho}\left(\frac{1}{m\rho}\right)^{\beta}\ones,$ $\vy_\beta^{(0)} = \zeros$.
\State $\vz^{(0)} = \exp(\epsilon/4)\ones,$ $\beta' = \frac{\epsilon/4}{(1+\beta)\log(mn\rho/\epsilon)}$, $h = \frac{\beta\beta'}{16\epsilon}$
\For{$k=1$ to $K=1 + \lceil 2/(h\epsilon) \rceil$}
\State $\vx^{(k)} = (\ones + \vz^{(k-1)})^{-1/\beta'}$
\State $\vz^{(k)} = \vz^{(k-1)} + \epsilon h \tgrad_r(\vx^{(k)})$
\State $\vy_\beta^{(k)} = \frac{k-1}{k}\vy_\beta^{(k-1)} + (\mA \vx^{(k)})^{1/\beta}/k$
\EndFor
\State\Return $(1+\epsilon)\vy_\beta^{(K)}$
\end{algorithmic}
\end{algorithm}
We start by bounding the initial gap.
\begin{proposition}\label{prop:covering-init-gap}
Let $h_0 = H_0 = 1$. Then: $H_0G_0 \leq 2(1+\beta)g_\beta(\vy^*_\beta).$
\end{proposition}
\begin{proof}
By the same arguments as in the proof of Proposition~\ref{prop:a<1-initial-conditions}, $\vx^{(1)}=\vx^{(0)},$ and hence $U_0 = f_r(\vx^{(0)}).$ Let $\vx^*_\beta = \argmin_{\vx \geq \zeros}f_r(\vx).$ Then, the initial gap can be expressed as:
$
H_0G_0 = \innp{\nabla f_r(\vx^{(0)}), \vx^{(0)}} - \psih(\vz^{(0)}) + \phi(\vx^*_\beta).
$

By the choice of $\vx^{(0)},$ $(\mA\vx^{(0)})^{1/\beta}\leq \ones,$ and, therefore, $\nabla f_r(\vx^{(0)})\leq \zeros.$ Thus, $H_0G_0 \leq  -\psih(\vz^{(0)}) + \phi(\vx^*_\beta).$ 
As $\beta'$ chosen here is smaller than the one from Section~\ref{sec:alpha<1}, it follows by the same argument as in the proof of Proposition~\ref{prop:a<1-init-gap} that $-\psih(\vz^{(0)}) = \frac{\beta'}{\eps(1-\beta')}\innp{\ones, (\vx^{(0)})}^{1-\beta'}\leq \frac{1}{2}\innp{\ones, \vx^{(0)}},$ which is at most $\frac{1}{2}(1+\beta)g_\beta(\vy^*_\beta),$ by the choice of the initial point $\vx^{(0)}$ and Proposition~\ref{prop:covering-opt-value}. It remains to bound $\phi(\vx^*_\beta) = \psi(\vx^*_\beta) - \big\langle{\nabla \psi(\vx^{(0)})+\tgrad_r(\vx^{(0)}), \vx^*_\beta}\big\rangle.$ By the definition of $\psi,$ $\psi(\vx^*_\beta)\leq 0$ and $\nabla_j\psi(\vx^{(0)}) = \frac{1}{\eps}(1- (x_j^{(0)})^{-\beta'})\geq -1/2$. By the definition of $f_r,$ $\tgrad_r(\vx^{(0)})\geq -\ones.$ Hence:
\begin{equation*}
- \innp{\nabla \psi(\vx^{(0)})+\tgrad_r(\vx^{(0)}), \vx^*_\beta} \leq \frac{3}{2}\innp{\ones, \vx^*_\beta} \leq \frac{3}{2}(1+\beta)g_\beta(\vy^*_\beta),
\end{equation*}
where the last inequality is by Proposition~\ref{prop:covering-bnd-for-x*b}.
\end{proof}
Since the analysis from Section~\ref{sec:alpha<1} can be applied in a straightforward way to ensure that after $\lceil 2/(h\epsilon)\rceil$ iterations we have $H_kG_k \leq \epsilon(1+\beta)g_\beta(\vy^*_\beta),$ what remains to show is that we can recover an approximate solution to~\eqref{eq:beta-fair-covering} from from this analysis. Define:
\begin{equation}\label{eq:cov-def-of-y}
\vy^{(k)} = (\mA\vx^{(k)})^{1/\beta}\quad \text{ and }\quad \vy^{(k)}_\beta = \frac{\sum_{\ell=1}^k \vy^{(\ell)}}{k}.
\end{equation}
Notice that this is consistent with the definition of $\vy^{(k)}_\beta$ from \textsc{FairCovering}. We are now ready to state and prove the main result from this section.
\begin{theorem}\label{thm:fair-covering}
The solution $\vy^{(k)}_\beta$ produced by \textsc{FairCovering} after $K = 1 + \lceil 2/(h\epsilon)\rceil = O(\frac{(1+\beta)\log(mn\rho)}{\beta\epsilon})$ iterations satisfies 
$$
\mA^T\vy^{(K)}_\beta \geq (1-\epsilon/2)\ones
\quad\text{ and } \quad
g_\beta(\vy_\beta^{(k)})- g_\beta(\vy^*_\beta) \leq 3\epsilon(1+\beta)g_\beta(\vy^*_\beta).
$$
\end{theorem}
\begin{proof}%
By the definition of $f_r$ and $\vy^{(k)},$ we have $f_r(\vx^{(k)}) - \innp{\nabla f_r(\vx^{(k)}), \vx^{(k)}} = - \sum_{i=1}^m \frac{(y_i^{(k)})^{1+\beta}}{1+\beta} = -g_\beta(\vy^{(k)})$. Hence, 
$$
L_k = \frac{-\sum_{\ell=0}^k h_\ell g_\beta(\vy^{(\ell)}) + \psih(\vz^{(k)}) -\phi(\vx^*_\beta)}{H_k}.
$$
As, by Proposition~\ref{prop:covering-init-gap} and the analysis from Section~\ref{sec:alpha<1} it must be $G_K \leq 2\epsilon (1+\beta)g_\beta(\vy^*_\beta),$ and by Lagrangian duality $U_K = f_r(\vx^{(K+1)})\geq -g_\beta(\vy^*_\beta)$, we have that $L_K = U_K-G_K \geq -(1+2\eps(1+\beta))g_\beta(\vy^*_\beta).$ As $\psih(\vz^{(k)})\leq 0$ and $\phi(\vx^*_\beta) \geq \psi(\vx^*_\beta) \geq - \frac{1}{2}\langle{\ones, \vx^*_\beta}\rangle = -\frac{1}{2}(1+\beta)g_\beta(\vy^*_\beta)$ (because $\nabla\psi(\vx^{(0)}) + \tgrad_r(\vx^{(0)})\geq \zeros$ and, by the choice of $\beta',$ $\psi(\vx^*_\beta) \geq - \frac{1}{2}\langle{\ones, \vx^*_\beta}\rangle$), we have that:
\begin{equation}\label{eq:covering-app-bnd-1}
\begin{aligned}
-\frac{\sum_{\ell=0}^K h_\ell g_\beta(\vy^{(\ell)})}{H_K} &\geq - (1+2\epsilon(1+\beta))g_\beta(\vy_\beta^*) - \frac{(1+\beta)g_\beta(\vy^*_\beta)}{2H_K}\\
&\geq - (1+(9\epsilon/4)(1+\beta))g_\beta(\vy_\beta^*). 
\end{aligned}
\end{equation}
Recall that $h_0 = 1$, $h_\ell = h$ for $\ell \geq 1$ and $H_K = \sum_{\ell=0}^K h_\ell =1+ Kh.$ As $g_\beta$ is convex, by the definition of $\vy_\beta^{(K)}$ and Jensen's inequality:
\begin{equation}\label{eq:covering-app-bnd-2}
\begin{aligned}
\frac{\sum_{\ell=0}^K h_\ell g_\beta(\vy^{(\ell)})}{H_K} &= \frac{1}{H_0}g_\beta(\vy^{(0)}) + \frac{h\sum_{\ell=1}^K  g_\beta(\vy^{(\ell)})}{1 + hK}\geq \frac{hK}{1+hK}g_\beta(\vy^{(K)}_\beta)\\
&\geq \frac{1}{1+\epsilon/2}g_\beta(\vy^{(K)}_\beta).
\end{aligned}
\end{equation}
Hence,  combining~\eqref{eq:covering-app-bnd-1} and \eqref{eq:covering-app-bnd-2}, $g_\beta(\vy^K_\beta)\leq (1+ 3\eps(1+\beta))g_\beta(\vy^*_\beta).$

It remains to show that $\vy^{(K)}_\beta$ is nearly-feasible. By its definition, $\vy^{(K)}_\beta\geq \zeros$. We claim first that it must be $\vz^{(k)}\geq -(\epsilon/2)\ones.$ Suppose not. Then $\psih(\vz^{(k)}) = - \frac{\beta'}{\epsilon(1-\beta')}\sum_{j=1}^n (1+z_j^{(k)})^{-\frac{1-\beta'}{\beta'}}\leq - \frac{\beta'}{\epsilon(1-\beta')} (1-\epsilon/2)^{- \frac{1-\beta'}{\beta'}}<<-H_K(1+\beta)g_\beta(\vy^*_\beta)$. As (from the argument above) $\phi_\beta(\vx^*_\beta)\geq -\frac{1}{2}(1+\beta)g_\beta(\vy^*_\beta)$, it follows that $L_K << -(1+\beta)g_\beta(\vy^*_\beta),$ which is a contradiction, as we have already shown that $L_K  \geq -(1+\eps(1+\beta))g_\beta(\vy^*_\beta).$ Thus, we have, $\forall j$, $z_j^{(k)}\geq -\epsilon/2$. By the definition of $\vz^{(k)}$:
\begin{align*}
1 + z_j^{(K)} 
&\leq 1 + \epsilon \sum_{\ell=1}^K h_\ell \tgradj_r(\vx^{(\ell)})
\leq 1 + \epsilon \sum_{\ell=1}^K h_\ell \nabla_j f_r(\vx^{(\ell)}). %
\end{align*} 
Recall that $\nabla_j f_r(\vx^{(\ell)}) = -1 + (\mA^T\vy^{(\ell)})_j$. Hence: 
\begin{align*}
\mA^T\vy^{(K)}_\beta = \frac{\sum_{\ell=1}^K \mA^T\vy^{(\ell)}}{K}\geq \vz^{(K)} + \epsilon h K \geq (1-\eps/2)\ones.
\end{align*}
\end{proof}
{Observe that Algorithm~\ref{algo:covering} returns the point $(1+\epsilon)\vy^{(K)}.$ This is to ensure that all of the covering constraints are satisfied. The approximation error in the statement of Theorem~\ref{thm:fair-covering} is then affected only by a factor $(1+\epsilon)^{1+\beta}.$}
%
%
%
\section{Conclusion}\label{sec:conclusion}

We presented efficient distributed algorithms for solving the class of $\alpha$-fair packing and covering problems on a relative scale. This class of problems contains the unfair case of packing and covering LPs, for which we obtain convergence times that match that of the best known packing and covering LP solvers~\cite{d-allen2014using,LP-jelena-lorenzo,d-wang2015unified}. Our results greatly improve upon the only known width-independent solver for the general $\alpha$-fair packing~\cite{marasevic2015fast}, both in terms of simplicity of the convergence analysis and in terms of the resulting convergence time.

\section*{Acknowledgements}
We thank Ken Clarkson for his useful comments and suggestions regarding the presentation of the paper.

\bibliographystyle{siamplain}
{%
\bibliography{references}
}
\appendix
\section{Omitted Proofs} 
\label{app:omitted-proofs-prelims}
\begin{proof}[Proof of Proposition~\ref{prop:regularization}]
The proof of the first part follows by solving:
\begin{align*}
\psi^*(\mA F_\alpha(\vx) - \ones) %
&= \max_{\vy \geq \zeros}\Big\{\innp{\mA F_\alpha(\vx), \vy} - \frac{1}{C^{\beta}}\sum_{i=1}^m \frac{{y_i}^{1+\beta}}{1+\beta}\Big\},
\end{align*}
which is solved for $y_i = C(\mA F_\alpha(\vx))_i^{1/\beta}.$ 

Let $\vx = F_\alpha^{-1}((1-\epsilon)\vx^*_\alpha)$. Then, $\forall i,$ $(\mA F_\alpha(\vx))_i \leq 1-\epsilon$ and thus:
$$
C(\mA F_\alpha(\vx))_i^{\frac{1+\beta}{\beta}} \leq (1-\epsilon/4)^{1/\beta} \leq \Big(\frac{\epsilon}{4mn\rho}\Big)^{\alpha +1}.
$$
{Hence:}
\begin{equation}\label{eq:reg-part-appx}
    {C\sum_{i=1}^m(\mA F_\alpha(\vx))_i^{\frac{1+\beta}{\beta}} \leq m \Big(\frac{\epsilon}{4mn\rho}\Big)^{\alpha +1} \leq \Big(\frac{\epsilon}{4n\rho}\Big)^{\alpha +1}.}
\end{equation}
{From Proposition~\ref{prop:packing-bounds-on-opt}, we have for $\alpha \neq 1$ that $(1-\alpha)f_\alpha(\vx^*_\alpha) \geq n (n\rho)^{\alpha -1} \geq \frac{1}{\rho}$. Thus, Eq.~\eqref{eq:reg-part-appx} implies that in this case $C\sum_{i=1}^m(\mA F_\alpha(\vx))_i^{\frac{1+\beta}{\beta}} \leq \frac{\epsilon}{4}(1-\alpha)f_\alpha(\vx^*_\alpha).$ For $\alpha = 1,$ we can simply use that $(\frac{\epsilon}{4n\rho})^{\alpha +1}\leq \frac{\eps n}{16}.$} 

As 
$
f_\alpha((1-\epsilon)\vx^*_\alpha) = (1-\epsilon)^{1-\alpha}f_\alpha(\vx_\alpha^*) \geq \big(1-\frac{3\epsilon(1-\alpha)}{2}\big)f_\alpha(\vx_\alpha^*)
$  
for $\alpha \neq 1$ and 
$ 
f_\alpha((1-\epsilon)\vx^*_\alpha) = n\log(1-\epsilon) + f_\alpha(\vx^*_\alpha) \geq -\frac{3}{2}\epsilon n + f_\alpha(\vx_\alpha^*),
$  
it follows that 
$$
f_r(\vx^*_r) \leq f_r(\vx) = -f_\alpha((1-\epsilon)\vx^*_\alpha) + \frac{\beta C}{1+\beta}\sum_{i=1}^m (\mA F_\alpha(\vx))_i^{\frac{1+\beta}{\beta}} \leq -f_\alpha(\vx^*_\alpha) + 2\epsilon_f.
$$

Finally, the \eqref{eq:packing-transformed}-feasibility of $\vx^*_r$ (and, by the change of variables, \eqref{eq:alpha-fair-packing}-feasiblity of $\vxh_r$) follows from Proposition~\ref{prop:feasibility} and Lemma~\ref{lemma:smoothness}. 
\end{proof}
\begin{proof}[Proof of Lemma~\ref{lemma:smoothness}]%
We will only prove the first part of the lemma, as he second part uses the same ideas.  
Writing a Taylor approximation of $f_r(\vx + \Gamma\vx)$, we have:
\begin{equation}\label{eq:taylor}
f_r(\vx + \Gamma\vx) = f_r(\vx) + \innp{\nabla f_r(\vx), \Gamma\vx} + \frac{1}{2}\innp{\nabla^2 f_r(\vx + t\Gamma\vx)\Gamma\vx, \Gamma\vx},
\end{equation}
for some $t\in [0, 1]$. The gradient and the Hessian of $f_r$ are given by:
\begin{align}
\nabla_j f_r(\vx) =& \frac{1}{1-\alpha}\Big(-1 + \sum_i A_{ij}{x_j}^{\frac{1}{1-\alpha}-1}C(\mA F_\alpha(\vx))_i^{{1}/{\beta}}\Big)\label{eq:grad-fr}\\
\nabla^2_{jk} f_r(\vx) =& \ones_{\{j = k \text{ and } \alpha \neq 0\}}\frac{\alpha}{(1-\alpha)^2}\sum_i A_{ij}{x_j}^{\frac{1}{1-\alpha}-2}C(\mA F_\alpha(\vx))_i^{{1}/{\beta}}\notag\\
&+ \frac{1/\beta}{(1-\alpha)^2}\sum_{i'} A_{i'j}A_{i'k}(x_jx_k)^{\frac{1}{1-\alpha}-1}C(\mA F_\alpha(\vx))_{i'}^{{1}/{\beta}-1}. \label{eq:hessian-f-r}
\end{align}
To have the control over the change in the function value, we want to enforce that the Hessian of $f_r$ does not change by more than a factor of two in one step. To do so, let $\gamma_m$ be the maximum (absolute) multiplicative update. Then, to have $\nabla^2_{jk} f_r(\vx + \Gamma\vx) \leq 2\nabla^2_{jk} f_r(\vx)$, it is sufficient to enforce: (i) $(1\pm\gamma_m)^{\frac{1}{1-\alpha}-2\pm \frac{1}{\beta(1-\alpha)}} \leq 2$ (from the first term in \eqref{eq:hessian-f-r}) and (ii) $(1\pm\gamma_m)^{\frac{2}{1-\alpha}-2 + \frac{1-\beta}{\beta(1-\alpha)}}\leq 2$ (from the second term in \eqref{eq:hessian-f-r}). Combining (i) and (ii), it is not hard to verify that it suffices to have: 
$
\gamma_m \leq \frac{\beta|1-\alpha|}{2(1+\alpha\beta)}. 
$ 

Assume from now on that $|\gamma_j| \leq \gamma_m \leq \frac{\beta|1-\alpha|}{2(1+\alpha\beta)}$, $\forall j$. Then, we have:
\begin{align}
\frac{1}{2}\innp{\nabla^2 f_r(\vx + t\Gamma\vx)\Gamma\vx, \Gamma\vx} %
\leq &  \sum_j\frac{\alpha}{(1-\alpha)^2}\sum_i {\gamma_j}^2 A_{ij}{x_j}^{\frac{1}{1-\alpha}}C(\mA F_\alpha(\vx))_i^{\frac{1}{\beta}}\notag\\
&+ \frac{1/\beta}{(1-\alpha)^2}\sum_{i'} C(\mA F_\alpha(\vx))_{i'}^{\frac{1}{\beta}-1} (\mA \Gamma\vx^{\frac{1}{1-\alpha}})_{i'}^2.\notag
\end{align}
Observe that, by Cauchy-Schwartz Inequality, 
$$(\mA \Gamma\vx^{\frac{1}{1-\alpha}})_{i'}^2 = (\sum_j A_{i'j}{x_j}^{\frac{1}{1-\alpha}}\gamma_j)^2 \leq 
(\mA F_\alpha(\vx))_{i'} \sum_j A_{i'j}{x_j}^{\frac{1}{1-\alpha}}{\gamma_j}^2.$$
Therefore, applying the last inequality and changing the order of summation:
\begin{equation}\label{eq:hessian-term}
\begin{aligned}
\frac{1}{2}\innp{\nabla^2 f_r(\vx + t\Gamma\vx)\Gamma\vx, \Gamma\vx} &%
&\leq \frac{1 + \alpha \beta}{\beta(1-\alpha)^2}\sum_j {\gamma_j}^2 {x_j}\left( (1-\alpha)\nabla_j f_r(\vx) + 1\right).
\end{aligned}
\end{equation}
Since 
$
\innp{\nabla f_r(\vx), \Gamma\vx} = \sum_j \gamma_j {x_j} \nabla_j f_r(\vx)
$ 
and $|\tgradj_r(\vx)| \leq 2\big|\frac{(1-\alpha)\nabla_j f_r(\vx)}{1 + (1-\alpha)\nabla_j f_r(\vx)}\big|$, choosing $\gamma_j = -\frac{c_j}{4}\cdot \frac{\beta(1-\alpha)}{1+\alpha\beta}\tgradj_r(\vx)$ and combining \eqref{eq:hessian-term} and \eqref{eq:taylor}:
\begin{align*}
f_r(\vx + \Gamma\vx) - f_r(\vx) &\leq -\frac{\beta(1-\alpha)}{1+\alpha\beta}\sum_{j} \frac{c_j}{4}\big(1-\frac{c_j}{2}\big)  x_j \nabla_j f_r(\vx) \tgradj_r(\vx)\\
&= \sum_{j=1}^n \Big(1 - \frac{c_j}{2}\Big)\gamma_j x_j \nabla f_r(\vx).
\end{align*}
\end{proof}

\end{document}